\documentclass[12pt]{article}

\usepackage[a4paper, tmargin=1in, lmargin=1in, rmargin=1in, bmargin=1in]{geometry}

\usepackage{setspace}
\usepackage{footmisc}
     
    \newlength{\myfootnotesep}
\usepackage{soul}
\usepackage{graphicx}
\usepackage{xcolor}
\usepackage{natbib}
\usepackage{amsmath} 
\usepackage{amsthm}
\usepackage{amsfonts}
\usepackage{amssymb} 
\usepackage{multirow}
\usepackage{dsfont}
\usepackage{tabularray}
\usepackage{booktabs}
\usepackage{float}
\usepackage{comment}
\usepackage{tikz} 
\usepackage{url} 
\usepackage{threeparttable}

\usepackage{relsize}
\DeclareMathSizes{12}{11}{8}{6}

\theoremstyle{plain}

\newtheorem{proposition}{Proposition}
\newtheorem{assumption}{Assumption}

\newtheorem{lemma}{Lemma}[section]

\usetikzlibrary{positioning, patterns, matrix, fit}

\tikzset{set/.style={draw,circle,inner sep=0pt,align=center}}

\title{Three’s a crowd: Identification challenges in the triple difference model with spillover effects
}
\author{Silvia De Nicolò\footnote{\texttt{silvia.denicolo@unibo.it}}, Beatrice Biondi, Mario Mazzocchi}
\date{Università di Bologna}

\begin{document}
\maketitle
	
\begin{abstract}
The paper studies identification in triple-difference designs when spillover effects contaminate one or more control groups. We show that, under conventional identifying assumptions, the triple-difference model fails to identify both the treatment effect and the spillover effect under such interference. To overcome this limitation, we propose an alternative specification, the double–triple-difference model, and explicitly formalize identifying assumptions and spillover structures required for consistent identification of both effects. We derive formal identification results and assess the performance of the proposed model through Monte Carlo simulations. An empirical application evaluating a Special Economic Zone in Italy is provided.\\
\\
 \textbf{Keywords}: Causal inference; Policy evaluation; Spillover effects; Triple difference.\\
\textbf{JEL codes}: C21, C18, C22, D04. \\
    \end{abstract}

\section{Introduction}
The estimation of policy treatment effects using triple-difference (TD) designs has become increasingly common in applied economics. A Scopus search for articles with the keyword ``triple difference'' in journals with ``economic(s)'' or ``econometric(s)'' in the title or abstract yields 327 documents between 2005 and 2025; 231 of these were published between 2020 and 2025. 

The growing popularity of the TD stems from its appealing features, which naturally relax the so-called parallel trends assumption commonly required in DiD frameworks, while maintaining a parsimonious and interpretable structure \citep{olden2022triple}. 
The TD approach extends the standard DiD setup by introducing an additional dimension of comparison. Treated observations are identified by the intersection of (i) a primary treatment partition (e.g., jurisdictions that adopt a policy) and (ii) a secondary partition (e.g., subpopulations eligible for the policy). As outlined by \citet{olden2022triple}, the TD framework allows for deviations from the parallel trends assumption provided that such deviations are constant between sets of a given partition. This softer requirement, usually named the \emph{parallel trend-in-trends} assumption, is sufficient because any common bias in two DiD estimators is differenced out when taking the third difference.

Other claimed advantages of the TD specifications relate to their ability to accommodate group-level heterogeneity in treatment effects \citep[e.g.][]{caron2025triple, Leung2021, Liu2020}, and control for spillover effects \citep{Bellou2013, yu2021evaluating, Lee2022}. Regarding the latter, estimation and interpretation of spillover effects through TD models is widespread in applied work, yet it has received limited formal treatment in the methodological literature. The review by \citet{olden2022triple} also endorses this feature, without providing an explicit formalization of the identifying assumptions in case of spillover effects. In contrast, other empirical studies are more cautious when dealing with TD designs in presence of spillover effects \citep[e.g.,][]{Curtis2018}. 

The recent proliferation of TD applications therefore calls for a more rigorous examination of the estimator's identification properties and underlying assumptions \citep{Roth2023}. Recent work has advanced our understanding of TD models in a range of environments as discussed hereinafter in Section \ref{sec:empirical_lit}. However, to the best of our knowledge, this emerging literature on TD models has not yet provided a formal analysis of identification issues in presence of spillover effects. Instead, methodological works on spillovers have largely developed within the DiD framework \citep[e.g.,][]{butts2021difference}. 

In this spirit, our paper formalizes the identifying assumption of the TD model in the case of spillover effects through a revisited potential outcome framework in line with \citet{huber2021framework}. Under the canonical regression specification of TD models, the average treatment effect on the treated (ATT) and the spillover effect are supposed to be captured by the triple-interaction coefficient, the TD estimand, and a lower-order interaction coefficient, respectively. We show that, once spillovers affect control units, these coefficients do neither identify the ATT, nor the spillover effect under conventional assumptions. This implies that TD models may yield biased and inconsistent estimates, potentially leading to incorrect substantive conclusions. 

To address this problem, we propose an alternative specification, the \textit{Double-Triple Difference} (DTD) model. The DTD model is designed to jointly identify the ATT and spillover effects, without imposing overly restrictive assumptions, such as multiple parallel trends. In doing so, we consider two specifications: (i) one in which the parallel trend-in-trends assumption holds unconditionally, and (ii) one in which it holds conditionally on observed covariates. We present the identifying assumptions and corresponding estimators for both specifications. The latter one extends the doubly-robust TD approach of \citet{ortiz2025better} to the DTD setting. 

We explore the implications of ignoring spillovers through an empirical application and two Monte Carlo simulations. The application studies the introduction of a Special Economic Zone (SEZ) in the port areas of Campania (Italy) in 2018. The SEZ provides tax credits for hiring and capital investment to firms located in eligible municipalities and operating in targeted, export-oriented sectors (e.g., automotive and food industry), with the aim of fostering exports, employment, and innovation. Spillovers to non-targeted sectors are plausible because the policy may affect upstream and complementary activities linked to seaborne exports (e.g., transport services and shipbuilding). Using firm-level panel data, we estimate the effect of the SEZ on employment exploiting two partitions: (i) sectoral, targeted versus non-targeted sectors; and (ii) geographical, firms in the Campania SEZ versus firms located in a similar area of Sicily not covered by a SEZ. Our findings indicate that the SEZ increased employment among treated firms by about 6.8\% on average and generated positive spillovers toward non-eligible sectors (3.8\%). We also show that the conventional TD specification underestimate the ATT when positive spillovers are present, as spillover-induced changes in a control group partially offset the true effect. 

In our first Monte Carlo simulation we adopt the same setting as in our empirical application, featuring panel data with multiple pre- and post-treatment periods under the standard unconditional parallel trend-in-trend assumption. Our second simulation considers the case in which the parallel trend-in-trend assumption holds only conditionally on observed covariates. This second simulation is based on a two-period design, drawing on the doubly-robust triple-difference framework proposed by \citet{ortiz2025better}. In both simulations, we compare the performance of the TD estimator with our proposed extension to the DTD framework. 

The paper is structured as follows. Section \ref{sec:empirical_lit} reviews the empirical literature on TD models, with a specific focus on studies that consider spillover effects, and motivates our analysis. Section \ref{sec::meth} presents the methodological framework and  outlines the standard TD model. Our proposed extension, the DTD model, is presented in Section \ref{sec::doubletriple}. Section \ref{sec::sim} reports the results of the two Monte Carlo simulation studies. Section \ref{sec:Application} applies our framework to evaluate the Campania SEZ policy. Finally, Section \ref{Sec:Conclusions} discusses the main findings and their implications for policy evaluation and methodological research.

\section{Triple-difference models and spillover effects in the literature}
\label{sec:empirical_lit}

Over the last few years, several theoretical contributions have clarified the conditions for identification, estimation, and inference in TD models across a range of settings, including multiple-period designs \citep{strezhnev2024group, akbari2025semiparametric, ortiz2025better}, staggered adoption \citep{strezhnev2023decomposing, ortiz2025better}, frameworks in which observed covariates play a central role in the credibility of identifying assumptions \citep{leventer2025conditional, ortiz2025better}, and settings with heterogeneous treatment effects \citep{akbari2024non, caron2025triple}. Empirically, heterogeneous treatment effects are also a common motivation for TD designs \citep{Xiong2025, Barwick2025, Barkowski2025, Premkumar2025, Paredes2025}. In addition, TD specifications are frequently employed as robustness checks following standard DiD and event-study analyses \citep[see e.g.][]{Reddig2024, Lipton2021}. By adding an additional control dimension relative to a double-difference specification, the intent is often to capture further unobserved trend components and to assess whether the implied ATT is robust to a richer set of control contrasts.\footnote{See \citet{olden2022triple} for a recent review of TD applications.}
\par Despite this large and growing literature, to the best of our knowledge no studies explicitly discuss the implications of TD designs when spillover effects are plausible. Within the DiD framework, only a few studies have examined this issue \citep{butts2021difference,huber2021framework, hettinger2025doubly, lee2025policy}, and, more generally, research on the topic remains limited in the causal panel data literature \citep{arkhangelsky2024causal}. To gauge how applied work treats this issue, we reviewed empirical studies across different fields that employ TD designs and in which spillovers are either explicitly mentioned or plausibly relevant. Three recurring practices emerge.

\paragraph{Practice 1: No-spillover as a maintained identifying assumption.}
A first set of studies acknowledges that TD identification hinges on the additional comparison dimension being unaffected by the policy, and therefore treats “no spillovers” as a maintained assumption backed mainly by institutional context. In \citet{McIntosh2008}, the application concerns whether microfinance innovations cushion households against localized shocks (e.g., disease outbreaks or political unrest), using outcomes such as borrowing, repayment, and savings. Microfinance products are introduced in some regions but not others. Within treated regions, adoption depends on surpassing a take-up threshold in a public election. The TD design then adds an additional dimension by comparing those who are willing to adopt (choosers) to those who are not (non-choosers), with chooser status in control regions proxied via a mock election that mimics the adoption process. The key identifying claim is that, absent the innovation, choosers and non-choosers would not affect one another. This rules out, for instance, informal risk-sharing or credit-market interactions that could transmit benefits from adopters to non-adopters and contaminate the comparison. Similarly, \citet{Lipton2021} studies whether expanding Medicaid coverage of preventive dental care for adults affects children’s dental care: a baseline DiD contrasts children of Medicaid-covered parents in states with versus without adult dental benefits, and a TD robustness check adds children of parents not covered by Medicaid as an extra comparison group. Interpreting the TD estimate requires that the expansion not spills over onto non-covered low-income families (e.g., via changes in provider participation, clinic capacity, or information), an assumption the paper states but does not directly test.

\paragraph{Practice 2: Spillovers acknowledged qualitatively.}
Other studies acknowledge that spillovers may be present and discuss their likely direction, using qualitative reasoning to sign the resulting bias in the ATT. For example, \citet{Bratberg2015} studies the impact of a financial shock on sickness absence among public-sector workers, exploiting a fiscal crisis affecting the public sector in some Norwegian municipalities in 2007. Their TD design uses municipalities not exposed to the shock as one placebo dimension and private-sector workers as another comparison group. The authors discuss a potential ``fear effect'' that could increase sickness absence even in unaffected municipalities, as well as a potential ``contagion effect'' affecting private-sector workers. On this basis, they argue that the ATT estimate is likely biased downward and thus conservative. No formal checks of the identifying assumptions are provided.

\paragraph{Practice 3: Lower-order interactions interpreted as capturing spillovers.}
A third set of studies explicitly discusses interference and estimates TD specifications with the goal of recovering both an ATT and a spillover component, typically by reading the relevant double interaction(s) as effects on comparison units. In the canonical application, \citet{Gruber1992} studies U.S. state mandates requiring health insurance coverage for maternity care. The TD design contrasts adopting versus non-adopting states and uses demographic groups with little direct exposure to maternity benefits (e.g., older women and men) as an additional comparison dimension. The evidence points to substantial cost shifting onto the wages of women of childbearing age and to limited changes in total labor input. In this setting, any post-mandate changes for the demographic comparison groups in adopting states are naturally interpreted as potential spillovers (e.g., broader wage-setting responses by employers). Several later applications adopt the same ATT and spillover reading in settings where spillovers are inherently plausible. \citet{Bellou2013} examines U.S. states’ adoption of vertically oriented driver’s licenses for under-21s, intended to ease age verification and curb underage smoking/drinking; their TD uses older teens in the same state (who typically keep horizontal IDs) as a within-state comparison group, so spillovers correspond to any policy-induced changes among older cohorts (peer effects, retailer behavior, or substitution in access). \citet{Zhang2016} studies Beijing’s urban-village removals and estimates a TD hedonic model using housing transactions: homes close to villages versus farther away, near removed versus non-removed villages, and before versus after removal, where the key double interaction is directly interpreted as the spatial ``spillover'' of removal on nearby formal-house prices. \citet{Lee2022} analyzes an expansion of dental coverage for the elderly in South Korea alongside contemporaneous improvements in chronic-disease coverage, using TD to isolate the incremental effect of dental coverage for those affected by the chronic-disease expansion; the corresponding lower-order interaction is discussed as capturing utilization responses that arise even absent the dental expansion. Finally, \citet{Bokhari2024} evaluates minimum alcohol unit pricing introduced in Wales in 2020 using product-level price and purchase data. They explicitly allow for substitution toward untreated products, more expensive alcoholic drinks not affected by the policy. The lower-order interaction on control products is treated as a spillover component, and the ATT is obtained by appropriately combining the triple and lower-order terms under multiple parallel trend assumptions.

The range of practices observed in the empirical literature points to a lack of consensus among practitioners. We therefore lay out transparent guidelines and propose an approach to properly identify the ATT under interference.

\section{Methodological framework}
\label{sec::meth}

We consider a population divided into two strata, one exposed to the policy, the \emph{treatment stratum} $\mathcal{S}_1$, and one not exposed, the \emph{placebo stratum} $\mathcal{S}_0$. For example, consider two countries: Country 1 and Country 0. A new anti-poverty policy is implemented in Country 1 (\( \mathcal{S}_1 \)), whereas Country 0 does not implement it (\( \mathcal{S}_0 \)). Within each stratum, units are categorized into two groups: individuals in \( \mathcal{T} \) are the policy \textit{target group}, whereas individuals in \( \mathcal{I} \) are not. To account for the stratum dimension, we append a subscript \(0\) or \(1\) to the group notation, indicating membership in stratum \(\mathcal{S}_0\) or \(\mathcal{S}_1\), respectively. The resulting partition is therefore \( \mathcal{S}_0 = \mathcal{T}_0 \cup \mathcal{I}_0 \) and \( \mathcal{S}_1 = \mathcal{T}_1 \cup \mathcal{I}_1 \). In our example, only specific individuals benefit from the anti-poverty policy, namely those in the lower income quantiles: these individuals are included in \( \mathcal{T}_1 \) in Country 1, and in \( \mathcal{T}_0 \) if they reside in Country 0, which does not implement such a policy.
Units are assumed to remain in the same group over time.
In this framework, we are interested in a policy outcome variable, $Y \in \mathbb{R}$, 
which is observed for the same population at two distinct time points, denoted by the index \( t \in \{0,1\} \). Specifically, $Y_{igst}$ is defined as the outcome for unit $i$, in group $g$ and stratum $s$, at time $t$. Let us also consider two dummy variables, $S$ and $G$, that take the value 1 if a unit belongs to the treatment stratum or target group, respectively, and 0 otherwise.
A binary treatment or policy is applied exclusively to units in the treatment stratum ($S=1$) belonging to the target group ($G=1$) during period $t=1$: the individual treatment variable is denoted with \( D \in \{0,1\} \). 

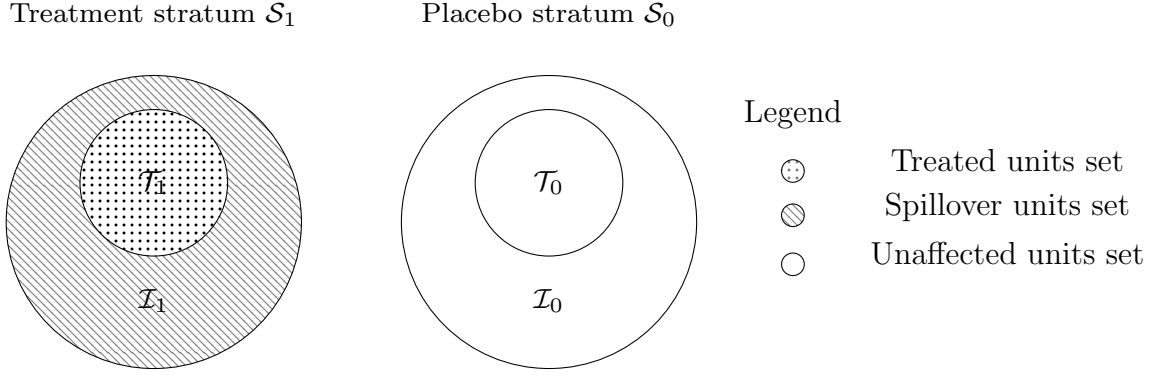
\begin{figure}
	\begin{tikzpicture}[scale=1.3, transform shape]
		\node[set,text width=3cm,pattern=north west lines, pattern color=gray] (leftbg) at (0,-0.4) {};
		\node[set,preaction={fill=white},pattern=dots,text width=1.5cm] (left) at (0,0) {};
		\node[above=0.7cm of left] {\scriptsize Treatment stratum $\mathcal{S}_1$};
		\node at (0,0) {\scriptsize$\mathcal{T}_1$}; 
		\node at (0,-1.2) {\scriptsize$\mathcal{I}_1$}; 
		
		\node[set,text width=3cm] (rightbg) at (4,-0.4) {};
		\node[set,text width=1.5cm] (right) at (4,0) {};
		\node[above=0.7cm of right] {\scriptsize Placebo stratum $\mathcal{S}_0$};
		\node at (4,0) {\scriptsize$\mathcal{T}_0$}; 
		\node at (4,-1.2) {\scriptsize$\mathcal{I}_0$}; 
		
		\matrix[draw=white, matrix of nodes, column sep=1mm] (m) at (8,0) {
			Legend \\
			\node [set,text width=3mm, pattern=dots, pattern color=gray, draw=black]{}; & Treated units set\\
			\node [set,text width=3mm, pattern=north west lines, pattern color=gray, draw=black]{}; & Spillover units set\\
			\node [set,text width=3mm, draw=black]{}; & Unaffected units set\\
		};
	\end{tikzpicture}
	\caption{Two strata, each comprising two groups. Treated units belong to group $\mathcal{T}_1$, and units affected by spillovers belong to group $\mathcal{I}_1$. Units in stratum $\mathcal{S}_0$ are completely unaffected, both directly and indirectly, by the policy.}
	\label{stratum_sutva}
\end{figure}
To account for spillover effects, we employ the potential outcome framework as has been utilized in studies related to spillover effects by, among others, \citet{forastiere2016identification} and \citet{huber2021framework}.
We define a spillover effect as any interference where the outcome of a unit is affected by the treatment status of another unit, thus not  depending exclusively on its own treatment status \citep{Hudgens2008}. Spillover effects violate the classical Stable Unit Treatment Value Assumption (SUTVA) \citep{frohlich2019impact}. 

For notational convenience, we express the potential outcome by distinguishing between the individual treatment status of a unit and its potential exposure to spillover.  Specifically, we denote the potential outcome of individual \( i \) in group \( g \) and stratum \( s \) at time $t$, when the individual treatment is assigned as \( D \in \{0,1\} \) and the spillover status is set to \( \tilde{D} \in \{0,1\} \), taking the value 1 if unit is exposed to spillovers, and 0 otherwise as:
\[
Y_{igst}(D_{igst} = d, \tilde{D}_{igst} = \tilde{d}).
\]
Therefore, we may define the following four potential outcomes:
\begin{enumerate}
	\item $Y_{igst}(1,1)$ represents the outcome when unit $i$ is treated and exposed to spillovers.
	\item $Y_{igst}(1,0)$ represents the outcome when unit $i$ is treated but not exposed to spillovers.
	\item $Y_{igst}(0,1)$ represents the outcome when unit $i$ is untreated but exposed to spillovers.
	\item $Y_{igst}(0,0)$ represents the outcome when unit $i$ is untreated and not exposed to spillovers.
\end{enumerate}

In line with \citet{huber2021framework}, we consider spillover effects to exist only within strata and not between strata, as summarized by Figure \ref{stratum_sutva}. Considering our initial example, spillovers may affect untreated units in higher-income quantiles within the treated country, while no spillovers occur in the placebo country.

\begin{assumption}[Stratum SUTVA]
\label{ass::sutva}
	The assumption of \textit{stratum SUTVA} implies that the potential outcome of a unit \( i \) in stratum \( s \) may be influenced by its treatment status and by the treatment status of units in the same stratum $s$ in which it is nested. Thus, SUTVA holds only across strata.
\end{assumption}

We also assume that there are no anticipatory effects for the four types of potential outcomes, which is a standard assumption ruling out any average effects of $D$ or $\tilde{D}$ on the outcome in the baseline period. Such effects could arise, for instance, if some units modified their behavior in the baseline period in anticipation of the upcoming individual treatment or exposure to spillovers. For the sake of simplicity, the notation $Y_{igst}$ is henceforth reduced to $Y_{it}$.

\begin{assumption}[No average anticipation effect]
\label{ass::anticipation}
	The following condition holds:  \begin{equation*}
		\mathbb{E}[Y_{i0}(d,\tilde{d}) - Y_{i0}(0, 0) \mid  S=s, G=g] = 0, \quad \forall s, g,
	\end{equation*}
	for all $d$, $\tilde{d} \in \lbrace 0,1\rbrace$.
\end{assumption}

In the presence of spillover effects, multiple definitions of the ATT arise. We begin by introducing the ATT in the absence of spillovers. For a generic stratum \( s \) and group \( g \), we define
\begin{equation*}
	ATT(g, s) = \mathbb{E}\!\left[ Y_{i1}(1,0) - Y_{i1}(0,0) \mid G=g, S=s \right].
\end{equation*}
This estimand captures the causal effect of treatment on a treated unit which is not affected by any spillover effect from other units. Alternatively, one may define an ATT that allows for spillover effects affecting treated units:
\begin{equation*}
	ATT^{\star}(g, s) = \mathbb{E}\!\left[ Y_{i1}(1,1) - Y_{i1}(0,0) \mid G=g, S=s \right].
\end{equation*}
The difference \( ATT^{\star}(g,s) - ATT(g,s) \) can be interpreted as the average spillover effect on treated units. Finally, it is useful to introduce the average spillover effect on the untreated (ASU), defined as
\begin{equation*}
	ASU(g, s) = \mathbb{E}\!\left[ Y_{i1}(0,1) - Y_{i1}(0,0) \mid G=g, S=s \right].
\end{equation*}
This estimand measures the impact of exposure to treated peers on units that do not themselves receive treatment. All these effects are local in nature, as they are defined conditionally on \( G \) and \( S \). More aggregate estimands can be obtained by marginalizing over \( G \), thereby recovering total effects within treated strata, as discussed in \citet{huber2021framework}.

For identification purposes, and in line with \citet{fiorini2024simple} and \citet{lee2025policy}, we impose the following assumption.
\begin{assumption}[No spillover effects on treated units]
	\label{ass::nospillover}
    The following condition holds:
	    \begin{equation*}
    	\mathbb{E}\!\left[ Y_{i1}(1,1) - Y_{i1}(1,0) \mid G=g, S=s \right] = 0,
		\quad \forall\, g, s.
	\end{equation*}	
\end{assumption}
\noindent
This assumption requires that once a unit receives treatment, its outcome is no longer influenced by spillover effects from other treated units. Under this assumption, spillovers operate only on untreated units, implying that \( ATT(g,s) = ATT^{\star}(g,s) \) for all \( g \) and \( s \). Henceforth, our quantities of interest are the ATT and the ASU, for which we provide identifying assumptions and estimation strategies.

\subsection{The triple-difference model}
\label{subsec::triple}

Consider the two-period TD specification in \citet{olden2022triple}:
\begin{align}
	\label{mod::3diff}
	Y_{it} \;=\;& \beta_0 + \beta_1 S_i + \beta_2 G_i + \beta_3 T_{t}
	+ \beta_4 (S_i \times G_i) + \beta_5 (G_i \times T_{t}) + \psi (S_i \times T_{t})+
	\nonumber\\
	&\quad + \delta (S_i \times G_i \times T_{t}) + \epsilon_{it}, \quad \forall i,t,
\end{align}
\noindent
where $T_{t}$ is a binary indicator equal to 1 if observation $(i,t)$ is in the post-policy period and 0 otherwise.\footnote{Equation \eqref{mod::3diff} can be equivalently written as a saturated three-way fixed effects regression, as discussed in \citet{ortiz2025better}.} The remaining notation follows our previous definitions.
The triple-interaction coefficient $\delta$ can be written as the difference of two DiD:
\begin{align*}
	\delta
	=\;& \Big(\mathbb{E}[Y_{i1}-Y_{i0}\mid S=1,G=1]-\mathbb{E}[Y_{i1}-Y_{i0}\mid S=1,G=0]\Big) \nonumber\\
	&-\Big(\mathbb{E}[Y_{i1}-Y_{i0}\mid S=0,G=1]-\mathbb{E}[Y_{i1}-Y_{i0}\mid S=0,G=0]\Big).
\end{align*}

The standard identifying assumption for $\delta$ in TD designs is the so-called \emph{parallel trend-in-trends} assumption \citep{Egami2023}\footnote{The same assumption is named ``parallel growth'' or ``common acceleration'' in \citet{frohlich2019impact}}. Compared with the DiD parallel trends assumption, it allows for non-parallel trends between treated and untreated groups within each stratum, but requires that deviations from parallel trends be the same across strata, as formalized in the following.

\begin{assumption}[Unconditional parallel trend-in-trends]
\label{ass::uptint}
	\begin{align*}
		\Delta^{TT}\;&\equiv\;\Big(\mathbb{E}\!\left[Y_{i1}(0,0)-Y_{i0}(0,0)\mid S=1,G=1\right]
		-\mathbb{E}\!\left[Y_{i1}(0,0)-Y_{i0}(0,0)\mid S=1,G=0\right]\Big)+ \\
		\qquad &-
		\Big(
		\mathbb{E}\!\left[Y_{i1}(0,0)-Y_{i0}(0,0)\mid S=0,G=1\right]
		-\mathbb{E}\!\left[Y_{i1}(0,0)-Y_{i0}(0,0)\mid S=0,G=0\right]
		\Big)
		=0.
	\end{align*}
	\vspace{-10mm}
\end{assumption}
\noindent
In this context, we derive formal identification expression for the triple-interaction parameter $\delta$ in presence of spillovers. Under stratum SUTVA (Assumption \ref{ass::sutva}), identification hinges on how outcomes of untreated units respond to spillover exposure, as stated in the following proposition. 

\begin{proposition}	\label{theo::triplediff_impact}
	Under Assumptions \ref{ass::sutva}, \ref{ass::anticipation}, \ref{ass::nospillover}, and \ref{ass::uptint}, the TD parameter $\delta$ satisfies
	\[
	\delta \;=\; ATT(S=1,G=1)\;-\;ASU(S=1,G=0).
	\]	Thus, it identifies the ATT if and only if the ASU is equal to zero.
\end{proposition}
\begin{proof}
	In the Supplementary Material.
\end{proof}

\noindent
Proposition \ref{theo::triplediff_impact} implies that, in presence of spillovers, the parameter $\delta$ in \eqref{mod::3diff} generally fails to identify the ATT and may result in biased estimates. In particular, when spillovers are positive for untreated units (a diffusion effect), $\delta$ understates the direct treatment effect. Conversely, when spillovers are negative (a displacement effect), $\delta$ overstates the ATT.

From a different perspective, spillovers can be viewed as an additional causal effect distinct from the direct treatment effect, with potentially different magnitude and sign. In this sense, our result relates to the empirical literature employing TD models to study heterogeneity across subgroups \citep{caron2025triple}, under different identifying assumptions.

We next consider identification of the spillover parameter. The coefficient $\psi$ in \eqref{mod::3diff} is sometimes interpreted as capturing the effect of the treatment on untreated units in the treated stratum after policy implementation \citep{olden2022triple}. This interpretation, however, requires stronger conditions than those needed to interpret $\delta$: identification of $\psi$ relies on a \emph{standard} DiD-type parallel trends assumption comparing untreated groups across strata. Specifically:
\begin{assumption}[Parallel trends for untreated groups]
	\label{ass::parallel_trend}
	\begin{align*}
		\Delta^{T}_{G=0}\;\equiv\;&
		\mathbb{E}\!\left[Y_{i1}(0,0)-Y_{i0}(0,0)\mid S=1,G=0\right]
		-
		\mathbb{E}\!\left[Y_{i1}(0,0)-Y_{i0}(0,0)\mid S=0,G=0\right]
		=0.
	\end{align*}
	\vspace{-10mm}
\end{assumption}
\noindent
The corresponding identification result for $\psi$ is stated in the following proposition.
\begin{proposition}
	\label{theo::triplediff_spillover}
	Under Assumptions \ref{ass::sutva}, \ref{ass::anticipation}, and \ref{ass::parallel_trend}, the spillover coefficient $\psi$ identifies the ASU:
	\[
	\psi \;=\; ASU(S=1,G=0).
	\]
\end{proposition}
\begin{proof}
	In the Supplementary Material.
\end{proof}

\noindent
Propositions \ref{theo::triplediff_impact} and \ref{theo::triplediff_spillover} are the first main identification result of our paper. The standard TD model does not, in general, identify the ATT in the presence of spillover effects, leading to biased estimates. Moreover, the spillover coefficient $\psi$ is identifiable only under a different and more stringent condition than parallel trend-in-trends, namely the classical DiD parallel trends assumption for untreated groups across strata. Hence, while TD designs may be more robust than standard DiD when the parallel trends assumptions fail, spillovers threaten identification of both the treatment and spillover effects.

We then expand the analysis to settings in which spillovers extend to the placebo stratum, thereby relaxing stratum SUTVA (Assumption \ref{ass::sutva}) and generating multiple, simultaneous spillovers across groups. More specifically, we consider two types of spillovers: one affecting those in the interference group in the treated stratum (as allowed by Assumption \ref{ass::sutva}), and a second spillover affecting units within a specific group of the placebo stratum. This yields two cases: Case A, where the spillover in the placebo stratum affects the target group, and Case B, where it affects units in the interference group. These cases are illustrated in Figures \textcolor{black}{S.1} and \textcolor{black}{S.2} of the Supplementary Material, and formal identification results are provided in \textcolor{black}{Section S.3}. We show that, in both cases, the TD model does not identify the ATT unless all spillover effects are zero. The identification problem disappears only in knife-edge configurations: in Case A, when the two spillovers exactly offset each other with equal and opposite effects, or in Case B, when they exactly reinforce each other with equal effects in the same direction.

Therefore, under spillover effects, ensuring unbiased ATT and ASU estimates raises complex identification issues. Suppose that practitioners can define
a placebo stratum that is unaffected, either directly or indirectly, by the policy, but cannot rule out spillovers in the treatment stratum. This is the baseline scenario  depicted in Figure \ref{stratum_sutva}, which is in line with Assumption \ref{ass::sutva}.
In this context, one potential strategy implies the estimation of the ATT and  ASU by means of two separate DiD parameters. The treatment effect can be identified by comparing group $\mathcal{T}_1$ (treated) with group $\mathcal{T}_0$ (control). Conversely, the spillover effect can be identified by comparing group $\mathcal{I}_1$ (exposed to spillovers) with group $\mathcal{I}_0$ (unexposed). However, this approach relies critically on the strong requirement that the parallel trends assumption condition holds in both comparisons. We therefore aim to relax such requirement and propose an estimation approach within a TD framework. The following section impose relevant assumptions and the formalization of our proposal.

\section{The double-triple difference model}
\label{sec::doubletriple}

Our goal is to propose an approach for estimating both direct treatment and spillover effects in the presence of spillovers, without relying on overly restrictive assumptions. To this end, we remain within a TD framework, which allows us to replace the classical DiD parallel trends condition with the weaker (and typically more plausible) parallel trend-in-trends assumption, while addressing identification issues discussed in Section \ref{sec::meth}.

\par Suppose we refine the partition within each stratum into three groups, so that
$\mathcal{S}_1=\mathcal{T}_1\cup\mathcal{I}_1\cup\mathcal{C}_1$ and
$\mathcal{S}_0=\mathcal{T}_0\cup\mathcal{I}_0\cup\mathcal{C}_0$.
Here, $\mathcal{T}$ denotes the policy target group; $\mathcal{I}$ denotes units that are not directly treated but may be exposed to spillovers (\textit{interference set}); and $\mathcal{C}$ denotes units that are neither directly treated nor plausibly affected by spillovers (\textit{pure control set}). By construction, the policy directly affects only units in $\mathcal{T}_1$, whereas -- under Assumption \ref{ass::sutva} -- spillovers affect only units in $\mathcal{I}_1$. This setup is illustrated in Figure \ref{fig::23diffscenario}. Let $G_i$ and $I_i$ be dummy variables equal to 1 if unit $i$ belongs to groups $\mathcal{T}$ and $\mathcal{I}$, respectively, and 0 otherwise. The full-control group $\mathcal{C}$ serves as the baseline and is characterized by $G_i=0$ and $I_i=0$.

We also define local versions of the ATT and ASU conditional on the newly introduced indicator variable $I$:
\begin{align*}
	ATT(g,s,j)
	&\equiv
	\mathbb{E}\!\left[ Y_{i1}(1,0)-Y_{i1}(0,0)\mid G=g,\;S=s,\;I=j\right],\\
	ASU(g,s,j)
	&\equiv
	\mathbb{E}\!\left[ Y_{i1}(0,1)-Y_{i1}(0,0)\mid G=g,\;S=s,\;I=j\right].
\end{align*}
\noindent
First, we impose an additional and crucial requirement: the existence (and observation) of units in the treated stratum that are untreated and unexposed to spillovers, in line with \citet{fiorini2024simple} and \citet{lee2025policy}.

\begin{assumption}[Observation of untreated units without spillovers]
	\label{ass::observationofI}
	The population sets $\mathcal{C}_1$ and $\mathcal{C}_0$ are non-empty and partially observed.
\end{assumption}

\noindent In this context, we define the \emph{double-triple difference} (DTD) model as
\begin{align}
	\label{mod::23diff}
	Y_{it}=\;&
	\beta_0 + \beta_1 S_i + \beta_2 T_{t} + \beta_3 G_i + \beta_4 I_i
	+ \beta_5 (S_i\times G_i) + \beta_6 (S_i\times T_{t}) + \beta_7 (G_i\times T_{t}) +
	\nonumber\\
	&\quad + \beta_8 (S_i\times I_i) + \beta_9 (I_i\times T_{t})
	+ \delta (S_i\times G_i\times T_{t})
	+ \psi (S_i\times I_i\times T_{t})
	+ \epsilon_{it}, \quad \forall i,t.
\end{align}
This specification embeds two TD parameters within a unified framework: the coefficient $\delta$ targets the ATT, while $\psi$ targets the spillover effect for the interference group.

\par The DTD coefficient $\delta$ can be written as
\begin{align}
	\label{eq:dtd_delta_def}
	\delta
	=\;&
	\Big(\mathbb{E}[Y_{i1}-Y_{i0}\mid S=1,G=1,I=0]-\mathbb{E}[Y_{i1}-Y_{i0}\mid S=1,G=0,I=0]\Big)+
	\nonumber\\
	&-
	\Big(\mathbb{E}[Y_{i1}-Y_{i0}\mid S=0,G=1,I=0]-\mathbb{E}[Y_{i1}-Y_{i0}\mid S=0,G=0,I=0]\Big).
\end{align}
\noindent
For $\delta$ to identify the ATT, we require a parallel trend-in-trends assumption for the comparison conditioning on $I=0$ (i.e., among units in the target and pure control groups).

\begin{assumption}[Unconditional parallel trend-in-trends between target and pure control groups]
	\label{ass:tinti0}
    \begin{align*}    
	\Delta^{TT}_{I=0}\;\equiv\;&\Big(\mathbb{E}[Y_{i1}(0,0)-Y_{i0}(0,0)\mid S=1,G=1,I=0]+\\
	&-\mathbb{E}[Y_{i1}(0,0)-Y_{i0}(0,0)\mid S=1,G=0,I=0]\Big)+\nonumber\\
	&-\Big(\mathbb{E}[Y_{i1}(0,0)-Y_{i0}(0,0)\mid S=0,G=1,I=0]+\\
	&-\mathbb{E}[Y_{i1}(0,0)-Y_{i0}(0,0)\mid S=0,G=0,I=0]\Big)=0.
	   \end{align*}
	\vspace{-10mm}
\end{assumption}
\noindent We can therefore state the identification result for $\delta$ as follows.

\begin{proposition}
	\label{prop:dtd_delta}
	Under Assumptions \ref{ass::sutva}, \ref{ass::anticipation}, \ref{ass::nospillover}, \ref{ass::observationofI}, and \ref{ass:tinti0}, the DTD parameter $\delta$ identifies
	\[
	\delta = ATT(S=1,G=1,I=0).
	\]
\end{proposition}
\begin{proof}
	In the Supplementary Material.
\end{proof}

\noindent We next turn to identification of the spillover parameter. The DTD specification includes the coefficient $\psi$, which is designed to recover the spillover effect under a parallel trend-in-trends assumption rather than the classical DiD parallel trends assumption.

\noindent The coefficient $\psi$ admits the representation
\begin{align*}
	\psi
	=\;&
	\Big(\mathbb{E}[Y_{i1}-Y_{i0}\mid S=1,G=0,I=1]-\mathbb{E}[Y_{i1}-Y_{i0}\mid S=1,G=0,I=0]\Big)+
	\nonumber\\
	&-
	\Big(\mathbb{E}[Y_{i1}-Y_{i0}\mid S=0,G=0,I=1]-\mathbb{E}[Y_{i1}-Y_{i0}\mid S=0,G=0,I=0]\Big).
\end{align*}
For $\psi$ to identify $ASU(S=1,G=0,I=1)$, we impose a parallel trend-in-trends assumption for the comparison conditioning on $G=0$.

\begin{assumption}[Unconditional parallel trend-in-trends between the interference and pure control groups]
	\label{ass:tinti1}
	\begin{align*}
	\Delta^{TT}_{G=0}\;\equiv\;&\Big(\mathbb{E}[Y_{i1}(0,0)-Y_{i0}(0,0)\mid S=1,G=0,I=1]+\\
	&-\mathbb{E}[Y_{i1}(0,0)-Y_{i0}(0,0)\mid S=1,G=0,I=0]\Big)+\nonumber\\
	&-\Big(	\mathbb{E}[Y_{i1}(0,0)-Y_{i0}(0,0)\mid S=0,G=0,I=1]+\\
	&-\mathbb{E}[Y_{i1}(0,0)-Y_{i0}(0,0)\mid S=0,G=0,I=0]	\Big).
\end{align*}
\vspace{-10mm}
\end{assumption}

\noindent Hence, we obtain the corresponding identification statement.

\begin{proposition}
	\label{prop:dtd_psi}
	Under Assumptions \ref{ass::sutva}, \ref{ass::anticipation}, \ref{ass::observationofI}, and \ref{ass:tinti1}, the spillover coefficient in the DTD model identifies
	\[
	\psi = ASU(S=1,G=0,I=1).
	\]
\end{proposition}
\begin{proof}
	In the Supplementary Material.
\end{proof}

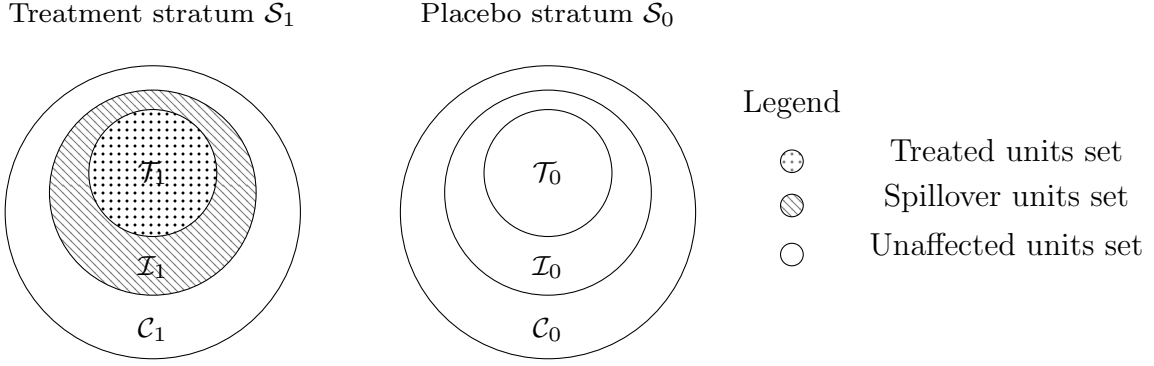
\begin{figure}
	\centering
	\begin{tikzpicture}[scale=1.3, transform shape]
		\node[set,text width=3cm]  (leftbg) at (0,-0.4) {};
		\node[set,text width=2.1cm, pattern=north west lines, pattern color=gray] (leftmid) at (0,-0.2) {};
		\node[set,preaction={fill=white},pattern=dots,text width=1.3cm] (leftin) at (0,0) {};
		\node[above=0.7cm of leftin] {\scriptsize Treatment stratum $\mathcal{S}_1$};
		
		\node at (0,0) {\scriptsize $\mathcal{T}_1$};
		\node at (0,-0.96) {\scriptsize $\mathcal{I}_1$};
		\node at (0,-1.6) {\scriptsize $\mathcal{C}_1$};
		
		\node[set,text width=3cm] (rightbg) at (4,-0.4) {};
		\node[set,text width=2.1cm] (rightmid) at (4,-0.2) {};
		\node[set,text width=1.3cm] (rightin) at (4,0) {};
		\node[above=0.7cm of rightin] {\scriptsize Placebo stratum $\mathcal{S}_0$};
		
		\node at (4,0) {\scriptsize $\mathcal{T}_0$};
		\node at (4,-0.96) {\scriptsize $\mathcal{I}_0$};
		\node at (4,-1.6) {\scriptsize $\mathcal{C}_0$};
		
		\matrix[draw=white, matrix of nodes, column sep=1mm] (m) at (8,0) {
			Legend \\
			\node [set,text width=3mm, pattern=dots, pattern color=gray, draw=black]{}; & Treated units set\\
			\node [set,text width=3mm, pattern=north west lines, pattern color=gray, draw=black]{}; & Spillover units set\\
			\node [set,text width=3mm, draw=black]{}; & Unaffected units set\\
		};
	\end{tikzpicture}
	\caption{Double-triple difference framework with two strata and three groups per stratum.}
	\label{fig::23diffscenario}
\end{figure}

\par Propositions \ref{prop:dtd_delta} and \ref{prop:dtd_psi} constitute our second main identification result: under Assumptions \ref{ass:tinti0} and \ref{ass:tinti1}, the DTD model identifies both the ATT and the ASU within a TD framework.

Note that when the parallel trend-in-trends assumptions hold unconditionally on control covariates -- as formalized in Assumptions \ref{ass::uptint}, \ref{ass:tinti0}, and \ref{ass:tinti1} -- the TD and DTD models in \eqref{mod::3diff} and \eqref{mod::23diff} can be consistently estimated by ordinary least squares (OLS), in line with \citet{ortiz2025better}. This is also the case in our empirical application. However, if parallel trend-in-trends holds only conditionally on control covariates, then naive OLS estimators become biased. This situation may arise, for instance, when group and/or stratum membership is correlated with observables that also affect outcome dynamics.
In such cases, consistent estimation requires covariate-adjusted procedures, such as regression adjustment, inverse probability weighting, or doubly-robust methods, as proposed by \citet{ortiz2025better}. We address this case by proposing a simple extension of our DTD difference framework to accommodate this situation.

\subsection{The double-triple difference framework in case of conditional parallel trend-in-trends}
\label{subsec:dr_dtdiff}

This subsection considers the case in which the parallel trend-in-trends assumptions hold only \emph{conditionally} on pre-treatment covariates. This may arise when selection into treatment is driven by observables, so that treatment and control groups differ systematically in both outcomes and observed characteristics. Albeit beyond the scope of our empirical application, a natural extension of our framework to this case can be implemented by separately estimating:
(a) the ATT, through a covariate-adjusted procedure applied to the sub-sample defined by the partitions $\mathcal{T}_1 \cup \mathcal{C}_1 \cup \mathcal{T}_0 \cup \mathcal{C}_0$, and
(b) the ASU, through a covariate-adjusted procedure applied to the sub-sample defined by the partitions $\mathcal{I}_1 \cup \mathcal{C}_1 \cup \mathcal{I}_0 \cup \mathcal{C}_0$.

\par We first introduce the extensions of Assumptions \ref{ass:tinti0} and \ref{ass:tinti1} defined conditionally on control covariates $X_i$, a $P$-dimensional vector. Those are more stringent cases compared to Assumptions \ref{ass:tinti0} and \ref{ass:tinti1}.

\begin{assumption}[Conditional parallel trend-in-trend between target and pure control groups]
\label{ass:tinti_c}
\begin{align*}
\Delta_{I=0,X}^{TT} &\equiv \mathbb{E}[Y_{i1} (0,0) - Y_{i0}(0,0) \mid S = 1, G = 1, I=0, X] +\\
&-\mathbb{E}[Y_{i1} (0,0) - Y_{i0}(0,0) \mid S = 1, G = 0, I=0,X]+\\
& - \Big( \mathbb{E}[Y_{i1} (0,0) - Y_{i0}(0,0)\mid S = 0, G = 1, I=0, X] +\\
&- \mathbb{E}[Y_{i1} (0,0) -  Y_{i0}(0,0) \mid S = 0, G = 0, I=0, X] \Big) = 0.
\end{align*}
\vspace{-5mm}
\end{assumption}

\begin{assumption}[Conditional parallel trend-in-trend between interference groups and pure control groups]
\label{ass:tinti1_c}
\begin{align*} \Delta_{G=0,X}^{TT} &\equiv \mathbb{E}[Y_{i1} (0,0) - Y_{i0}(0,0) \mid S = 1, G = 0, I=1, X] +\\&-\mathbb{E}[Y_{i1} (0,0) - Y_{i0}(0,0) \mid S = 1, G = 0, I=0, X]+\\
& - \Big( \mathbb{E}[Y_{i1} (0,0) - Y_{i0}(0,0)\mid S = 0, G = 0, I=1, X] +\\&- \mathbb{E}[Y_{i1} (0,0) -  Y_{i0}(0,0) \mid S = 0, G = 0, I=0, X] \Big) = 0.
\end{align*}
\vspace{-5mm}
\end{assumption}
In this context, we extend our DTD approach to the doubly-robust TD framework as proposed by \citet{ortiz2025better}. The TD estimand is defined as follows
\begin{align}\label{eq::dr_att}
\delta_{\texttt{dr}} &= \mathbb{E} \Bigg[ 
\Big( w^{S=1,G=1,I=0}_{i} - w^{S=1,G=1,I=0}_{i1,1}(X) \Big) 
\Big( Y_{i1} - Y_{i0} - m_i^{S=1,G=0,I=0}(X) \Big) \Bigg]+ \nonumber \\
&\quad + \mathbb{E} \Bigg[ 
\Big( w^{S=1,G=1,I=0}_{i} - w^{S=1,G=1,I=0}_{i0,1}(X) \Big) 
\Big( Y_{i1} - Y_{i1} - m_i^{S=0,G=1,I=0}(X) \Big)\Bigg]+ \nonumber \\
&\quad - \mathbb{E} \Bigg[ 
\Big( w^{S=1,G=1,I=0}_{i} - w^{S=1,G=1,I=0}_{i0,0}(X) \Big) 
\Big( Y_{i1} - Y_{i0} - m_i^{S=0,G=0,I=0}(X) \Big)\Bigg], \quad \forall i, 
\end{align}
with $m_i^{S=s,G=g,I=0}(X)=\mathbb{E}[Y_{i1}-Y_{i0}|S=s,G=g,I=0,X]$ determined by an outcome regression model and the estimated weights $w$ given by
\[
w^{S=1,G=1,I=0}_{i} = 
\frac{\mathbf{1}_i\{S=1, G=1,I=0\}}{\mathbb{E}[\mathbf{1}_i\{S=1, G=1,I=0\}]},
\]

\[
w^{S=1,G=1,I=0}_{isg}(X) = 
\frac{\mathbf{1}_i\{S=s, G=g,I=0\} \cdot \frac{p^{S=1,G=1,I=0}_{s,g}(X)}{(1-p^{S=1,G=1,I=0}_{s,g}(X))}}
{\mathbb{E} \left[ \mathbf{1}_i\{S=s, G=g,I=0\} \frac{p^{S=1,G=1,I=0}_{s,g}(X)}{(1-p^{S=1,G=1,I=0}_{s,g}(X))} \right]}.
\]
Note that, for generic strata $s$ and $s'$ and generic groups $g$ and $g'$, $$p^{S=s,G=g,I=0}_{s',g'}(X) = \frac{\mathbb{P}[S=s,G=g,I=0|X]}{\mathbb{P}[S=s,G=g,I=0|X]+\mathbb{P}[S=s',G=g',I=0|X]},$$ is determined by a generalized propensity score model following \citet{lechner2002program}.

\par A doubly-robust spillover parameter, denoted by $\phi_{\texttt{dr}}$, is constructed analogously to Equation \eqref{eq::dr_att} by replacing (i) $G=g$ with $G=0$ (for all $g$), (ii) $I=0$ with $I=j$, and (iii) redefining the weights accordingly: $w^{S=1,G=0,I=1}_{isj}(X)$ instead of $w^{S=1,G=1,I=0}_{isg}(X)$. Below we provide the identification statements for both parameters.

\begin{proposition}
	Under Assumptions \ref{ass::sutva}, \ref{ass::anticipation}, \ref{ass::nospillover}, \ref{ass::observationofI}, and \ref{ass:tinti_c}, the doubly-robust DTD parameter $\delta_{\texttt{dr}}$ defined in Equation \eqref{eq::dr_att} identifies
	\[
	\delta_{\texttt{dr}} = ATT(S=1,G=1,I=0).
	\]
\end{proposition}
\begin{proof}
	In the Supplementary Material.
\end{proof}

\begin{proposition}
	Under Assumptions \ref{ass::sutva}, \ref{ass::anticipation}, \ref{ass::observationofI}, and \ref{ass:tinti1_c}, the doubly-robust spillover parameter $\phi_{\texttt{dr}}$ identifies
	\[
	\phi_{\texttt{dr}} = ASU(S=1,G=0,I=1).
	\]
\end{proposition}
\begin{proof}
	In the Supplementary Material.
\end{proof}

\par These constitute the third main result of our paper: under conditional parallel trend-in-trends, $\delta_{\texttt{dr}}$ and $\phi_{\texttt{dr}}$ identify the ATT and ASU, respectively. Estimators of both $\delta_{\texttt{dr}}$ and $\phi_{\texttt{dr}}$ can be obtained by substituting $m^{S=s,\cdot,\cdot}(X)$ and $p^{S=s,\cdot,\cdot}(X)$ with predictions from their working models, and by replacing expectations with sample analogues. For further details, see \citet{ortiz2025better} and the implementation procedure is enclosed in the R package \texttt{triplediff}.

\par The TD and DTD estimators of Subsections \ref{subsec::triple} and \ref{sec::doubletriple}, together with the doubly-robust estimators introduced here under conditional parallel trend-in-trends, are evaluated through two simulation studies in Section \ref{sec::sim}. The former is implemented using OLS, whereas the latter relies on the doubly-robust estimators of \citet{ortiz2025better} to adjust for control covariates, using OLS for outcome regressions and logistic regression for propensity score models as working models.

\section{Monte Carlo Simulations}
\label{sec::sim}

We evaluate the frequentist properties of various estimators under both the TD and DTD frameworks in the presence of spillovers.
The first simulation study replicates the conditions of our empirical application, where panel data include multiple pre-treatment and post-treatment periods and the unconditional parallel trend-in-trends assumptions hold. In this setting, the TD and DTD setups correspond to those described in Sections \ref{subsec::triple} and \ref{sec::doubletriple}, with a straightforward extension to the panel-data context.
The second simulation study considers a two-period design and focuses on scenarios in which observed covariates are required for conditional parallel trend-in-trends assumptions to hold. We compare the doubly-robust TD estimator proposed by \citet{ortiz2025better} with its extension to the DTD framework introduced in Section \ref{subsec:dr_dtdiff}, accounting for spillovers. This simulation largely follows the design of \citet{ortiz2025better} and is implemented using the R package \texttt{triplediff}.

\subsection{Panel data and unconditional parallel trend-in-trends}

The data-generating process (DGP) is designed to closely mirror our empirical application in a panel-data setting. In particular, the ATT for the treated group is set to a known value, allowing for an assessment of estimator bias, error, and coverage under different model specifications and spillover scenarios.

Consider a set of $N=2{,}000$ observations repeatedly measured over time, with periods indexed by $t = 1, \dots, 10$. Note that, in this panel structure, the temporal treatment indicator is no longer binary and unit-specific. Instead, we use $T_t$ (for all $i$), which equals 1 in periods in which treatment is in place and 0 otherwise. In our setup, $T_t = 1$ for $t\geq 6$. Half of the sample units are randomly assigned to the treatment stratum and the other half to the placebo stratum. Within each stratum, half of the units are further assigned to the eligible group and the remaining half to the control group. Consequently, each sub-sample defined by the interaction between stratum and group consists of $500$ units. Moreover, within the control group of each stratum, we randomly select a sub-sample of units that may be subject to spillover effects, i.e., units that may indirectly benefit from (or be harmed by) the policy. We consider two sizes for the interference group, representing $10\%$ and $50\%$ of the overall control group ($\mathcal{I} \cup \mathcal{C}$). The dummy variables indicating stratum, group, and interference status are denoted by $S$, $G$, and $I$, consistently with the setup of Section \ref{sec::meth}.

We consider the following baseline DGP, corresponding to a standard panel-data extension of a three-way fixed effects specification:
\begin{align*}
	Y^{\text{SUTVA}}_{it} = \beta_i + \beta_t + \beta_{G,t} G_i + \beta_{S,t} S_i + \delta\, S_i \times G_i \times T_{t}+\varepsilon_{it}, \quad \forall i,t.
\end{align*}
Here, $\beta_i$ are unit fixed effects, $\beta_t$ are time fixed effects, and $\beta_{G,t}$ and $\beta_{S,t}$ capture group- and stratum-specific time effects. The parameter $\delta$ is the ATT, and $\varepsilon_{it}$ is an error term.

We generate $\beta_i$ as independent draws from $\mathcal{N}(\mu_u, \sigma_u^2)$. The parameters $\beta_t$, $\beta_{G,t}$, and $\beta_{S,t}$ are generated independently as random walks: $\beta_{\bullet,t}\sim \mathcal{N}(\beta_{\bullet,t-1},\sigma_t^2)$, with initial value $\beta_{\bullet,1}\sim \mathcal{N}(0,\sigma_t^2)$. Finally, $\varepsilon_{it}$ is drawn from $\mathcal{N}(0,\sigma_\epsilon^2)$. Hyperparameter values are extrapolated from the application and set to $\mu_u=0.90$, $\sigma_u=1$, $\sigma_t=0.05$, and $\sigma_\epsilon=0.50$. Lastly, $\delta$ is set to 0.20 (higher than in the application) to ensure a stronger signal.

The outcome $Y^{\text{SUTVA}}_{it}$ is generated under SUTVA, meaning that each unit's outcome is affected solely by its own treatment status and not by treatment assigned to other units. We then introduce spillovers through a set of scenarios that depart from this baseline:

\begin{itemize}
	\item \underline{Scenario 1}: This scenario is consistent with stratum SUTVA (Assumption~\ref{ass::sutva}), i.e., a unit's outcome may be influenced by the treatment status of other units \emph{within the same stratum}. Spillovers affect only the interference group in the treated stratum:
	\begin{align*}
		Y^{\text{S1}}_{it} = Y^{\text{SUTVA}}_{it} + \psi_1\, S_i \times I_i \times T_t, \quad \forall i,t.
	\end{align*}
	We assess estimator performance for different magnitudes of $\psi_1$, specifically 0.05, 0.10, and 0.20, corresponding to Scenarios 1.0, 1.1, and 1.2, respectively.
	\item \underline{Scenario 2}: This scenario relaxes stratum SUTVA by allowing spillovers to affect interference groups in both the treatment and placebo strata:
	\begin{align*}
		Y^{\text{S2}}_{it} = Y^{\text{SUTVA}}_{it} + \psi_1\, S_i \times I_i \times T_t + \psi_2\, (1-S_i) \times I_i \times T_t, \quad \forall i,t.
	\end{align*}
	In this scenario, $\psi_1$ is set to 0.10, while $\psi_2$ takes values $-0.10$ and $0.10$, defining Scenarios 2.0 and 2.1, respectively. Both TD and DTD models may be misspecified in this setting.
\end{itemize}

We conduct $K=1{,}000$ Monte Carlo iterations. For each scenario, we estimate both the TD and DTD models using their three-way fixed effects specifications. Estimation is carried out via OLS. For each iteration $k$, we store the impact estimate $\hat{\delta}^{(k)}$, the spillover estimate (when defined) $\hat{\psi}^{(k)}$, and their estimated clustered standard errors $\widehat{\text{se}}(\hat{\delta})$ and $\widehat{\text{se}}(\hat{\psi})$. For each model, scenario, parameter type, and interference-group size, we compute average bias, mean squared error (MSE), and 95\% coverage as follows:
\begin{align}\label{sim::performance}
	&\text{Bias} = \frac{1}{K}\sum_{k=1}^K \left(\hat{\bullet}^{(k)}-\bullet\right),
	\quad
	\text{MSE} = \frac{1}{K}\sum_{k=1}^K \left(\hat{\bullet}^{(k)}-\bullet\right)^2; \\
	&\text{Coverage}_{95\%} = \frac{1}{K} \sum_{k=1}^{K}
	\mathds{1}\!\left(
	\hat{\bullet}^{(k)}-1.96 \cdot \widehat{\text{se}}\!\left(\hat{\bullet}\right)^{(k)}
	\leq \bullet \leq
	\hat{\bullet}^{(k)}+1.96 \cdot \widehat{\text{se}}\!\left(\hat{\bullet}\right)^{(k)}
	\right),
	\quad \text{with} \quad \bullet \in \lbrace \delta, \psi \rbrace. \nonumber
\end{align}

\begin{table}[ht]
	\centering
    \begin{threeparttable}
    \caption{Average bias, MSE, and 95\% coverage of TD and DTD estimators with multiple time periods and unconditional parallel trend-in-trends -- results are displayed for different interference groups sizes and scenarios.}
    	\label{tab::sim1}
	\small
	\begin{tabular}{lll|rrr|rrr}
		\toprule
Scenario & Model & Type &
\multicolumn{6}{c}{Interference group size} \\
\cmidrule(lr){4-9}
& & &
\multicolumn{3}{c}{10\%} &
\multicolumn{3}{c}{50\%} \\
\cmidrule(lr){4-6}\cmidrule(lr){7-9}
& & &
Bias & MSE & Cov &
Bias & MSE & Cov \\
\midrule
		SUTVA        & triple        & ATT        & -0.002 & 0.001 & 0.944 & -0.002 & 0.001 & 0.944 \\ 
		SUTVA        & double-triple & ATT        & -0.002 & 0.001 & 0.944 & -0.002 & 0.001 & 0.936 \\ 
		SUTVA        & double-triple & spillover  & 0.000  & 0.043 & 0.946 & -0.001 & 0.041 & 0.948 \\ 
		\hline
		scenario 1.0 & triple        & ATT        & -0.007 & 0.001 & 0.939 & -0.027 & 0.002 & 0.866 \\ 
		scenario 1.0 & double-triple & ATT        & -0.002 & 0.001 & 0.944 & -0.002 & 0.001 & 0.936 \\ 
		scenario 1.0 & double-triple & spillover  & 0.000 & 0.025 & 0.946 & -0.001 & 0.024 & 0.948 \\ 
		\hline
		scenario 1.1 & triple        & ATT        & -0.012 & 0.001 & 0.925 & -0.052 & 0.004 & 0.605 \\ 
		scenario 1.1 & double-triple & ATT        & -0.002 & 0.001 & 0.944 & -0.002 & 0.001 & 0.936 \\ 
		scenario 1.1 & double-triple & spillover  & 0.000 & 0.013 & 0.946 & -0.001 & 0.011 & 0.948 \\ 
		\hline
		scenario 1.2 & triple        & ATT        & -0.022 & 0.001 & 0.893 & -0.102 & 0.011 & 0.095 \\ 
		scenario 1.2 & double-triple & ATT        & -0.002 & 0.001 & 0.944 & -0.002 & 0.001 & 0.936 \\ 
		scenario 1.2 & double-triple & spillover  & 0.000 & 0.003 & 0.946 & -0.001 & 0.001 & 0.948 \\ 
		\hline
		scenario 2.0 & triple        & ATT        & -0.022 & 0.001 & 0.891 & -0.102 & 0.011 & 0.094 \\ 
		scenario 2.0 & double-triple & ATT        & -0.012 & 0.001 & 0.930 & -0.052 & 0.004 & 0.672 \\ 
		scenario 2.0 & double-triple & spillover  & 0.000 & 0.013 & 0.946 & -0.001 & 0.011 & 0.948 \\ 
		\hline
		scenario 2.1 & triple        & ATT        & -0.002 & 0.001 & 0.944 & -0.002 & 0.001 & 0.944 \\ 
		scenario 2.1 & double-triple & ATT        & 0.008 & 0.001 & 0.937 & 0.048 & 0.004 & 0.725 \\ 
		scenario 2.1 & double-triple & spillover  & 0.000 & 0.013 & 0.946 & -0.001 & 0.011 & 0.948 \\ 
		\bottomrule
	\end{tabular}
	\begin{tablenotes}
    \scriptsize \item \textit{Notes:} Monte Carlo simulation with $2,000$ obs. and $1,000$ iterations; $10$ periods considered, treatment happens in period $7$; two strata of equal size, divided equally into target and control groups ($N=500$ each group); we consider two sizes for the interference group, representing 10\% and 50\% of the control groups in each stratum. Scenario SUTVA assumes SUTVA for all observations; scenario 1 assumes stratum SUTVA (Assumption \ref{ass::sutva}); scenario 2 considers multiple spillovers at once. The ATT parameter ($\delta$) is set to $0.2$ in all scenarios. Values of spillover in each scenario: 1.0: $\psi_1=0.05$; 1.1: $\psi_1=0.1$; 1.2: $\psi_1=0.2$; 2.0: $\psi_1=0.1$ and $\psi_2=-0.1$; 2.1: $\psi_1=0.1$ and $\psi_2=0.1$. For scenario 2, we only show results regarding spillover $\psi_1$.
  \end{tablenotes}
\end{threeparttable}
\end{table}

The results for the performance measures in \eqref{sim::performance}, by model, scenario, parameter type, and interference-group size, are reported in Table \ref{tab::sim1}.
Under SUTVA, the TD and DTD estimators display similar performance in terms of bias, MSE, and coverage. In the class of Scenario~1 subcases, where a single spillover affects the interference group within the treated stratum, the performance of the TD estimator deteriorates substantially due to non-identification. This results in considerable bias, increased MSE, and low coverage. As expected, the magnitude of the bias depends on two factors: (a) the size of $\psi_1$, as evidenced by the deterioration from Scenario~1.0 to 1.2 (with $\psi_1$ increasing from 0.05 to 0.20); and (b) the proportion of units the in interference group relative to the pure control group (from $10\%$ to $50\%$). In contrast, the DTD estimator maintains good performance across all Scenario~1 subcases and for both interference-group sizes.

Performance deteriorates further in Scenario~2, where spillovers also affect the interference group within the placebo stratum, implying misspecification for both TD and DTD models. In particular, when the two spillovers have opposite signs, the performance of the TD estimator worsens further. For instance, comparing Scenario~1.1 ($\psi_1 = 0.10$, $\psi_2 = 0$) with Scenario~2.0 ($\psi_1 = 0.10$, $\psi_2 = -0.10$), and focusing on the case of $50\%$ spillover size, the absolute bias increases from 0.052 to 0.102, leading to a higher MSE and severely reduced coverage (9.4\%). When considering the DTD estimator the absolute bias rises from 0.002 to 0.052, since it accounts for only one of the two spillovers. Conversely, in the knife-edge configuration of Scenario~2.1, where spillover effects are identical in size and sign ($\psi_1 = 0.10$, $\psi_2 = 0.10$), the identification issue affecting the TD estimator cancels out, yielding perfect identification. However, the DTD estimator displays a residual bias (0.048), as it accounts for only one of the two spillovers and compensation is not warranted. Both patterns are consistent with Proposition~S.2 in the Supplementary Material.

\begin{figure}[!t]
	\centering
	\includegraphics[width=\textwidth]{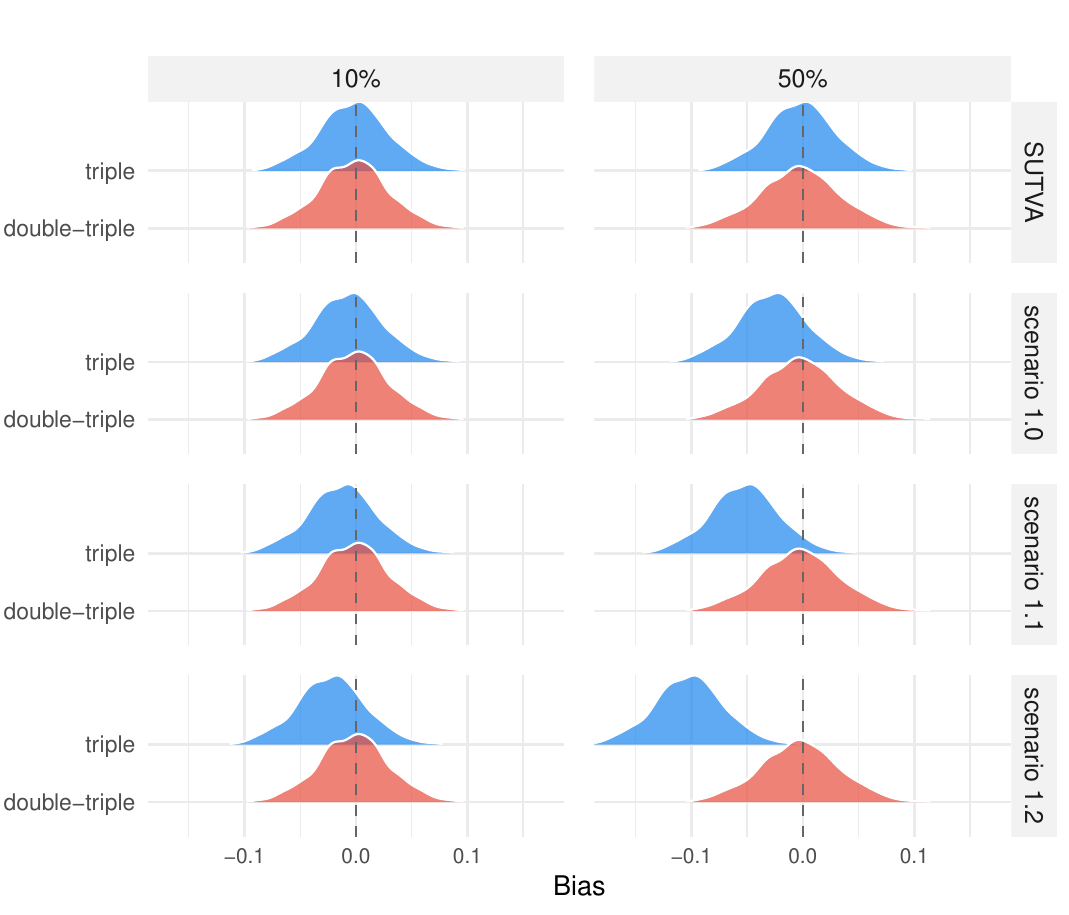}
	\caption{Bias of TD and DTD estimators for scenarios SUTVA, 1.0, 1.1, 1.2 at varying interference sets sizes (10\% and 50\% of the original control sets).}
	\label{fig:esempio}
\end{figure}

\begin{figure}[h]
	\centering
	\includegraphics[width=\textwidth]{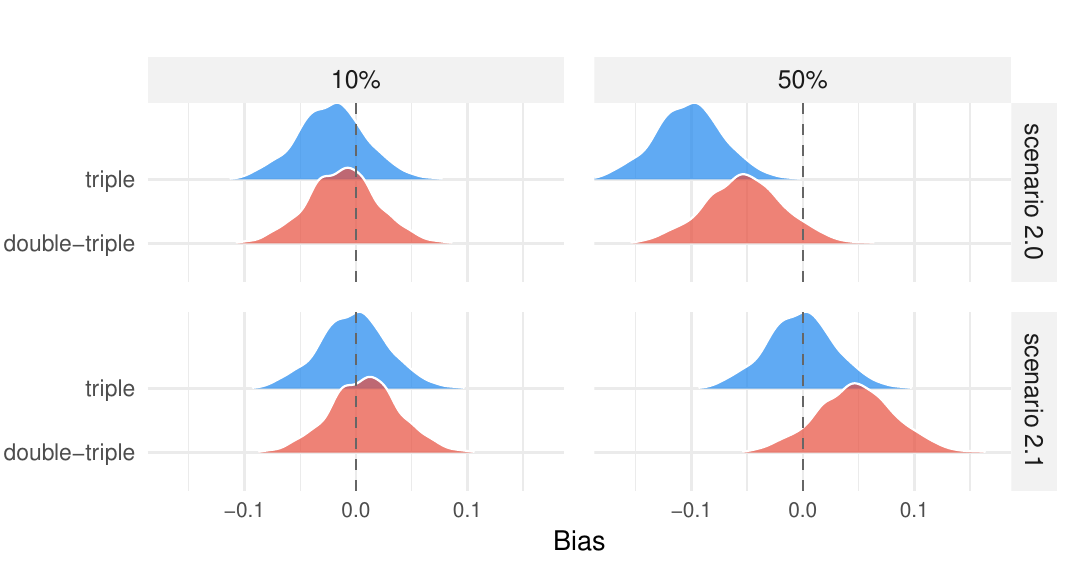}
	\caption{Bias of TD and DTD estimators for scenarios 2.0 and 2.1 (with multiple-spillover) 
    at varying interference sets sizes (10\% and 50\% of the original control sets).}
	\label{fig:esempio2}
\end{figure}

\subsection{Two-period data and conditional parallel trend-in-trend}

\begin{table}[h]
	\centering
        \begin{threeparttable}
    	\caption{Average bias, MSE, and 95\% coverage of TD and DTD estimators with two time periods and conditional parallel trend-in-trends -- results are displayed for different scenarios and sample sizes.}
        	\label{tab::sim2}
	\begin{tabular}{rlllrrr}
  \hline
  N & Scenario & Model & Type & Av. Bias & Av. MSE & Coverage \\
  \hline
  2,000 & SUTVA & triple & ATT & -0.001 & 0.017 & 0.954 \\
  2,000 & SUTVA & double-triple & ATT & -0.007 & 0.028 & 0.938 \\
  \cline{2-7}
  2,000 & spillover & triple & ATT & -12.495 & 156.338 & 0.000 \\
  2,000 & spillover & double-triple & ATT & -0.007 & 0.028 & 0.938 \\
  2,000 & spillover & double-triple & spillover & -0.007 & 0.037 & 0.946 \\
  \hline
  5,000 & SUTVA & triple & ATT & -0.001 & 0.007 & 0.954 \\
  5,000 & SUTVA & double-triple & ATT & -0.007 & 0.011 & 0.956 \\
  \cline{2-7}
  5,000 & spillover & triple & ATT & -12.487 & 156.001 & 0.000 \\
  5,000 & spillover & double-triple & ATT & -0.007 & 0.011 & 0.956 \\
  5,000 & spillover & double-triple & spillover & -0.011 & 0.015 & 0.944 \\
  \hline
  10,000 & SUTVA & triple & ATT & -0.003 & 0.004 & 0.948 \\
  10,000 & SUTVA & double-triple & ATT & -0.003 & 0.005 & 0.956 \\
  \cline{2-7}
  10,000 & spillover & triple & ATT & -12.517 & 156.715 & 0.000 \\
  10,000 & spillover & double-triple & ATT & -0.003 & 0.005 & 0.956 \\
  10,000 & spillover & double-triple & spillover & 0.001 & 0.008 & 0.952 \\
  \hline
\end{tabular}

	\begin{tablenotes}
    \scriptsize \item \textit{Notes:} Monte Carlo simulation with $500$ iterations. The interference group size is set to 50\% of the control groups in each stratum. SUTVA Scenario assumes SUTVA for all observations; spillover scenario relaxes it. The ATT parameter ($\delta$) is set to $50$ in all scenarios. Values of spillover ($\psi$) is set to $25$ in the spillover scenario.
  \end{tablenotes}
\end{threeparttable}
\end{table}
\begin{figure}[!htbp]
	\centering
	\includegraphics[width=\textwidth]{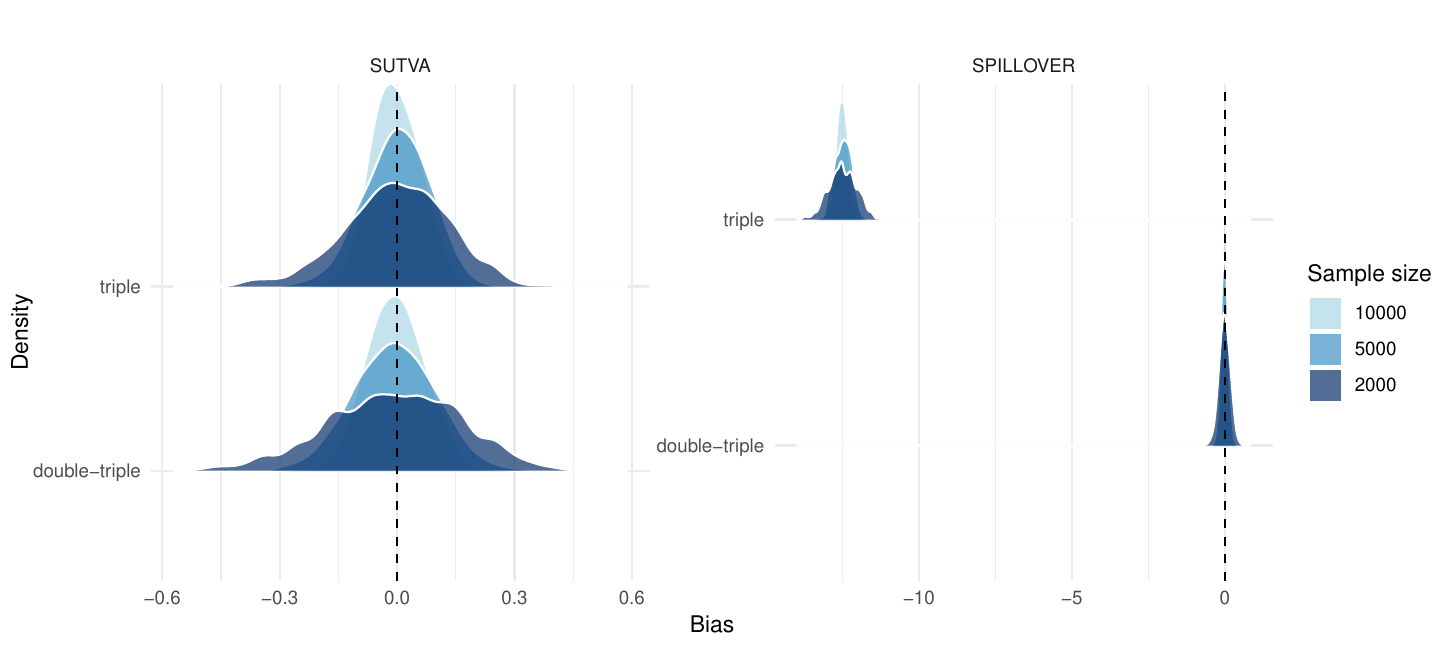}
	\caption{Bias of the doubly-robust TD and DTD estimators under SUTVA and spillovers at varying sample sizes.}
	\label{fig:sim2}
\end{figure}

We conducted a second simulation study to assess the performance of the doubly-robust TD estimator in the presence of spillovers, and to evaluate an extension of this estimation strategy that builds on the conceptual framework of our proposed DTD model under conditional parallel trend-in-trend assumptions. This simulation design largely overlaps with the one in \citet{ortiz2025better}.

Let us consider the basic two-period setup and notation of Section \ref{sec::meth}. The Monte Carlo simulations generate data for a two-period TD setup with unit-level covariates. Each unit $i$ has four covariates, $X_i=(X_{i,1},X_{i,2},X_{i,3},X_{i,4})$, drawn independently as detailed in Section S.9 of the Supplementary Material. Units are assigned to treatment and control groups across the two periods depending on these covariates, and outcomes are constructed to incorporate both covariate effects and group -- period interactions.
For each sub-group, assignment probabilities are defined as follows:
\begin{align*}
	\mathbb{P}[S=s,G=g\mid X]=p^{S=s,G=g}(X)
	=
	\frac{\exp\!\big(f^{ps}_{S=s,G=g}(X)\big)}{\sum_{(s,g)\in\mathcal{S}}\exp\!\big(f^{ps}_{S=s,G=g}(X)\big)}, \quad \forall s,g
\end{align*}
where $f^{ps}_{S=s,G=g}(X)$ is a linear predictor defined in Section S.9 of the Supplementary Material. Each unit is assigned to a subgroup according to these probabilities. Assignment to the interference group is implemented by randomly sampling half of the units assigned to sub-groups denoted by $G=0$.

The outcome under SUTVA is defined as
\begin{align}
	Y_{it}^{\text{SUTVA}}
	&=
	f^{reg}(X_i,S_i)
	+
	f^{reg}(X_i,S_i)\times T_{it}
	+
	\nu_i(X_i,S_i,G_i)
	+
	\delta\, S_i\times G_i\times T_{it}
	+
	\epsilon_{it}, \quad \forall i,t,
	\label{sim::sutva}
\end{align}
where $f^{reg}(X_i,S_i)$ is a linear predictor with coefficients that vary across strata $S$, and $\nu_i(X_i,S_i,G_i)$ adds additional time-invariant heterogeneity; precise definitions are provided in Section S.9 of the Supplementary Material. Lastly, $\epsilon_{it}$ (for all $t$) are independent standard normal random variables. We set $\delta=50$.

A second outcome specification relaxes SUTVA by allowing for spillovers, in line with Assumption \ref{ass::sutva}:
\begin{align}
	Y_{it}^{\text{spillover}}
	&=
	Y_{it}^{\text{SUTVA}}+
	\psi\, S_i\times I_i\times T_{it}, \quad \forall i,t,
	\label{sim::spillover}
\end{align}
with $\psi=25$. The number of Monte Carlo iterations is set to $K=500$. We evaluate performance in terms of bias, MSE, and coverage, using the measures in Equation \eqref{sim::performance}, for both the doubly-robust TD and DTD estimators under the DGP of SUTVA scenario defined in Equation \eqref{sim::sutva} and under the spillover scenario defined in Equation \eqref{sim::spillover}. We consider different overall sample sizes $N\in\lbrace 2{,}000, 5{,}000, 10{,}000\rbrace$.

Results are reported in Table \ref{tab::sim2}. Under SUTVA, the doubly-robust TD and DTD estimators perform comparably. However, in the spillover scenario, the performance of the TD estimator deteriorates markedly due to non-identifiability, leading to substantial bias, higher MSE, and low (or even zero) coverage. This is evident across all sample sizes. In contrast, the doubly-robust DTD estimator performs well in both scenarios. A slight deterioration is observed for smaller sample sizes, which is expected since, being applied to sub-samples, the DTD estimator requires larger overall sample sizes than the TD estimator for comparable precision.

Figure \ref{fig:sim2} illustrates the bias distribution under both scenarios. As the sample size increases, the distribution becomes more concentrated around its peak. Moreover, the substantial bias associated with the TD estimator is evident. This bias corresponds to $-\psi\times 0.5$, where $0.5$ is the fraction of the control group affected by spillovers. This result reinforces the underlying non-identifiability issue.

\section{Application to Special Economic Zones in Italy} \label{sec:Application}


We now turn to our empirical application, which examines the effects of tax incentives targeted at specific industries within geographically defined areas of an Italian region. This policy environment naturally generates multiple plausible control groups and thus provides a particularly suitable setting to illustrate both the presence and the empirical relevance of spillover effects. Specifically, we study the introduction of a Special Economic Zone (SEZ) in the southern Italian region of Campania.
Spillover effects induced by industrial policies have been discussed by \citet{Cerqua2017}, who acknowledge and address violations of SUTVA by combining matching methods with a DiD design. Like ours, their identification strategy rests on the availability of a control group that is not affected by spillovers.   
\subsection{Background}

A national regulatory framework for the establishment of SEZs was introduced in 2018 with the objective of fostering firm development in Southern Italy through tax incentives and administrative simplification.\footnote{Decree of the Presidency of the Council of Ministers (DPCM) 25 January 2018, n.\ 12, \url{https://www.gazzettaufficiale.it/eli/id/2018/2/26/18G00033/sg}.} 
Within this framework, regional governments in the South were allowed to establish SEZs, provided that the designated areas included at least one port integrated into the Trans-European Transport Network. Eligible regions could submit a strategic plan specifying (i) the municipalities included, (ii) the targeted sectors, and (iii) the set of measures available to eligible firms. The Campania SEZ was instituted in May 2018.\footnote{See, e.g., the institutional summary reporting the DPCM of 11 May 2018 establishing the SEZ Campania: \url{https://www.impresainungiorno.gov.it/route/zes?cod=campania}.} 
Sicily was instituted later, in July 2020, with the associated tax credit becoming operational only subsequently.\footnote{For the establishment of the Sicilian SEZs via DPCM(s) of 22 July 2020, see \url{https://www.impresainungiorno.gov.it/web/l-impresa-e-la-pa-centrale/zes-sicilia}. For implementing provisions on the investment tax credit, see the Revenue Agency communication model updated by provvedimento of 9 March 2021: \url{https://www.informazionefiscale.it/credito-imposta-investimenti-Mezzogiorno-Zes-modello-comunicazione-aggiornato}.}
\par The Campania SEZ includes 37 municipalities out of 550 in the region. Incentives primarily take the form of tax credits for capital and R\&D investments and, in addition, hiring-related measures. In our empirical application, we examine the impact of the SEZ on employment outcomes. Hiring incentives are targeted at disadvantaged groups (e.g.\ youth, the long-term unemployed, women, and individuals with disabilities) and are subject to restrictions intended to promote net job creation: eligible hires must not hold a permanent contract at the time of hiring nor have held one in the previous six months, and only newly created permanent positions qualify for the subsidy. These features plausibly mitigate concerns about job displacement through simple relabeling of existing employment relationships. 
\par The strategic plan places particular emphasis on sectors with relatively high export potential that could benefit from proximity to the ports and inland terminal infrastructure, including automotive, processed foods, apparel, machinery, electronics, wood products, and furniture. At the time of implementation, firms operating in these industries accounted for approximately 93\% of Campania's seaborne exports and employed about 99{,}000 workers (roughly one-third of regional manufacturing employment).\footnote{SEZ Campania strategic plan, \url{https://www.agenziacoesione.gov.it/wp-content/uploads/2019/09/Piano-Strategico-Campania.pdf}.}
\par This context naturally suggests using non-eligible sectors as a control group in the empirical analysis. Sectoral eligibility is constrained by EU State-aid rules, which explicitly exclude several industries from accessing national or regional incentives. These include agriculture, fisheries, coal, steel, shipbuilding, synthetic fibres, transport and related infrastructure, energy generation and distribution, financial services, head-office activities, and consulting. In addition, the regional strategic plan does not consider other sectors -- such as retail trade, telecommunications, and real estate services -- as eligible.
\par There are clear channels through which spillovers may affect employment in sectors not directly targeted by the policy. Given the explicit objective of promoting seaborne exports, one might reasonably anticipate spillovers into non-eligible shipbuilding and other maritime-related activities. More broadly, sectors such as construction, transport services (including maritime transport), and waste management may respond to SEZ-induced investment and the associated expansion in local economic activity. Supporting this view, evidence from the Apulia SEZ documents positive spillovers into logistics and transport despite their exclusion from direct policy benefits \citep{Bergantino2025}. In contrast, sectors with limited or negligible export orientation -- such as raw agricultural production (largely destined to food processing and domestic markets), public administration, retail trade, and some business services -- are less likely to experience such effects, at least in the short run. In our analysis, we leverage this distinction between likely and unlikely spillover-exposed sectors to illustrate the biases that can arise when indirect effects are ignored.
\par A second dimension of our quasi-experimental design is geographic. Our goal is to identify areas that are comparable to treated municipalities in economic and geographic characteristics, yet unlikely to be affected by the Campania SEZ. Nearby municipalities within Campania may be similar along many dimensions, but they are unsuitable controls due to potential spillovers, negative through labor mobility (albeit mitigated by the incentive design), or positive through complementarities and local demand effects. A more suitable comparison group can be drawn from municipalities in other regions that were designated as SEZs only later, and therefore were not yet treated at the time Campania implemented its SEZ. This choice is also consistent with the institutional evolution of place-based policies in Italy, as the SEZ framework built upon pre-existing targeted measures for disadvantaged areas. Accordingly, our second control dimension comprises municipalities included in the Sicily SEZ, which was formally instituted in 2020 but became operational only later. The municipalities designated as part of the Sicilian SEZ therefore provide a credible control group in the pre-activation period, while also strengthening the plausibility of the counterfactual by focusing on areas with similar economic profiles and prospective policy relevance.
\par Our empirical strategy implements a TD design along three dimensions: (i) eligible versus non-eligible industries; (ii) municipalities included in the Campania SEZ versus municipalities designated as SEZ areas in Sicily only in 2020; and (iii) time, distinguishing between pre- and post-policy periods. By drawing a clear distinction between sectors that are extremely unlikely to be affected by spillovers and sectors that may plausibly be exposed to indirect effects, we examine how estimates differ between the canonical TD specification and a DTD model that explicitly allows for spillover effects. For both approaches, we discuss the identifying assumptions and compare the baseline specification to a three-way fixed effects formulation.

\subsection{Data}

The data for our empirical application are drawn from AIDA (Orbis for Italy), accessed via the Bureau van Dijk platform. AIDA is a comprehensive firm-level panel database providing detailed financial and operational information on Italian companies. It includes balance sheets, income statements, and employment data, and covers firms across a wide range of industries, with up to ten years of historical data available. \citet{Cerqua2017} also rely on AIDA data to estimate employment and turnover effects of firm subsidies.\footnote{Further information about the dataset and access instructions are available at:
	\url{https://www.moodys.com/web/en/us/capabilities/company-reference-data/orbis/aida-orbis-for-italy.html}}
\par The final sample results from a cleaning process aimed at reducing noise and ensuring reliable measurement of the main variables. We excluded firms reporting zero revenues in all years, and firms that were in the process of liquidation during the period of analysis, hence not eligible for the hiring incentives. The resulting unbalanced panel consists of 87{,}631 unique firms observed between 2013 and 2020, comprising 454{,}985 observations. 
\par Table \ref{tab:desc_stats} reports descriptive statistics on the number of included firms and average employment in the pre-policy (2013-2017) and post-policy (2018-2020) periods, by region, targeted-sector status, and spillover exposure. Within control sectors, we further distinguish between those potentially exposed and those unlikely to be affected by spillovers.

\begin{table}[htbp]
	\centering
	\caption{Number of firms and average employment in the overall sample and key subsamples (pre-SEZ: 2013-17; post-SEZ: 2018-20)}
	\label{tab:desc_stats}
	\small
	\begin{tabular}{l r cc cc}
		\toprule
		& \multicolumn{1}{c}{N. firms} & \multicolumn{2}{c}{Pre} & \multicolumn{2}{c}{Post} \\
		\cmidrule(lr){3-4}\cmidrule(lr){5-6}
		&  & Mean & (SE) & Mean & (SE) \\
		\midrule
		\multicolumn{6}{l}{\textit{All sectors}} \\
		Campania    & 46{,}873 & 8.76 & (0.18) & 8.35 & (0.19) \\
		Sicily      & 40{,}758 & 7.95 & (0.13) & 7.64 & (0.37) \\
		Difference  & 87{,}631 & 0.81 & (0.23) & 0.71 & (0.42) \\
		\addlinespace
		
		\multicolumn{6}{l}{\textit{Target sectors}} \\
		Campania    & 5{,}174  & 13.76 & (0.40) & 13.52 & (0.48) \\
		Sicily      & 6{,}604  & 8.74  & (0.19) & 7.87  & (0.18) \\
		Difference  & 11{,}778 & 5.02  & (0.44) & 5.65  & (0.51) \\
		\addlinespace
		
		\multicolumn{6}{l}{\textit{Control sectors (all)}} \\
		Campania    & 35{,}438 & 7.89 & (0.20) & 7.48 & (0.20) \\
		Sicily      & 40{,}069 & 7.83 & (0.15) & 7.61 & (0.43) \\
		Difference  & 75{,}507 & 0.06 & (0.25) & -0.13 & (0.47) \\
		\addlinespace
		
		\multicolumn{6}{l}{\textit{Control sectors (interference group)}} \\
		Campania    & 13{,}642 & 10.14 & (0.38) & 10.04 & (0.45) \\
		Sicily      & 14{,}510 & 8.14  & (0.25) & 7.80  & (0.25) \\
		Difference  & 28{,}152 & 2.00  & (0.45) & 2.24  & (0.51) \\
		\addlinespace
		
		\multicolumn{6}{l}{\textit{Control sectors (pure control group)}} \\
		Campania    & 21{,}796 & 6.66 & (0.24) & 6.05 & (0.19) \\
		Sicily      & 25{,}559 & 7.64 & (0.19) & 7.49 & (0.67) \\
		Difference  & 47{,}355 & -0.98 & (0.30) & -1.44 & (0.69) \\
		\bottomrule
	\end{tabular}
	
	\vspace{0.75ex}
	\begin{minipage}{0.95\linewidth}
		\scriptsize
        \textit{Source:} Our processing on AIDA data, Bureau van Dijk. \\
		\textit{Notes:} Means and standard errors of the mean (in parentheses) are computed using all firm-year observations. Standard errors for differences are computed assuming unequal variances between groups. 
	\end{minipage}
\end{table}

\subsection{Results}

\begin{figure}[ht!]
	\centering
	\includegraphics[width=\textwidth]{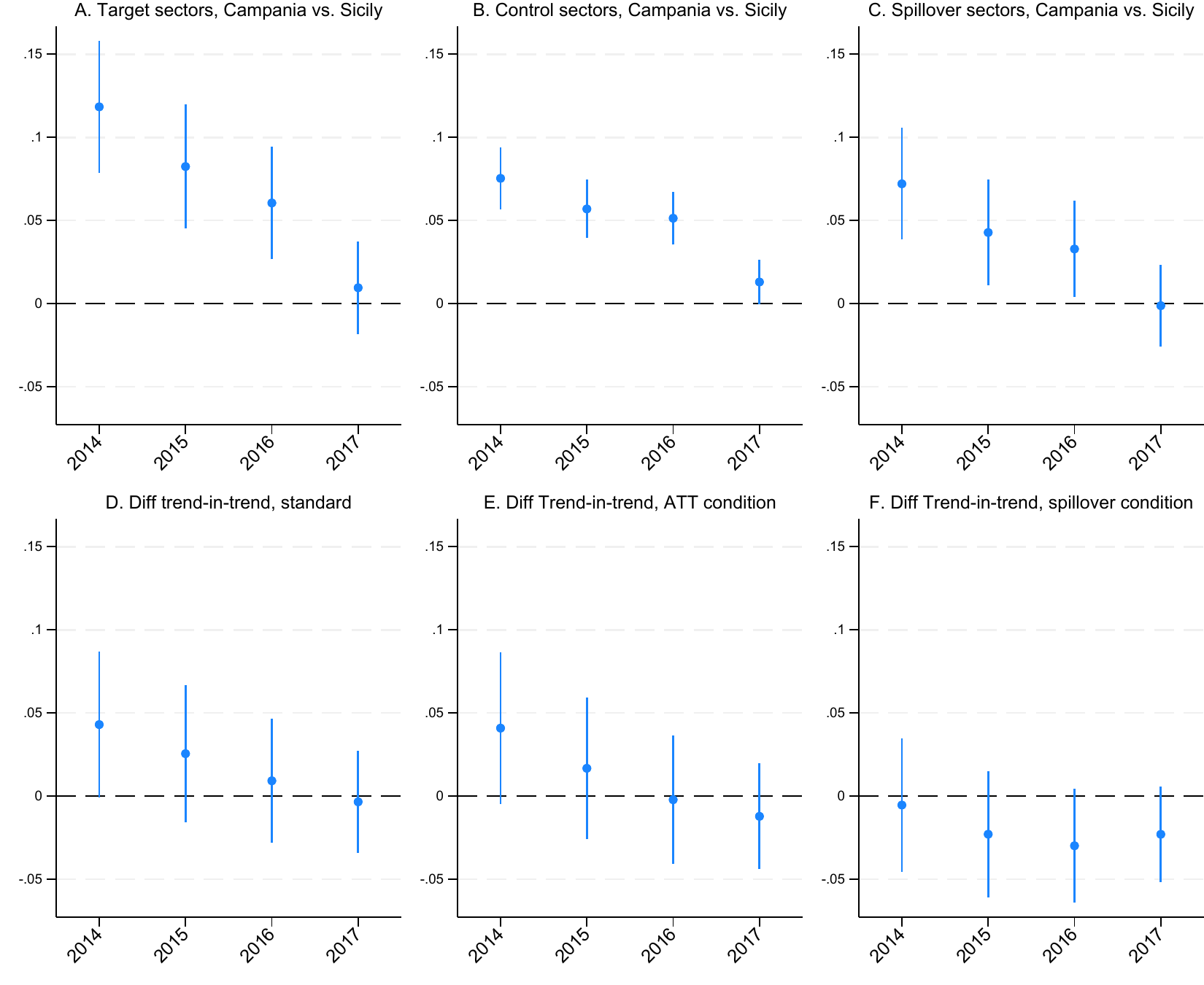}
	\caption{Tests on parallel trends and trend-in-trends assumptions in pre-policy years (2013-2017).}
	\label{fig:evstudy}
\end{figure}

We begin by assessing the plausibility of the identifying assumptions by testing for parallel trends and parallel trend-in-trends in the pre-policy window (2013-2017). We follow the event study specification in \citet{Miller2023}, while considering pre-policy years only. The estimating equations and the testing procedures are detailed in Section S.10 in the Supplementary Material. More specifically, Equation S.10 provides the standard event study test for the DiD parallel trends assumption, while Equation S.11 tests the parallel trend-in-trends assumption (Assumption \ref{ass::uptint}) required by the TD model. Note that our DTD framework relies on two parallel trend-in-trends assumptions, Assumption \ref{ass:tinti0} to identify the ATT and Assumption \ref{ass:tinti1} to identify the ASU. We test these assumptions by estimating Equation S.11 on the two following sub-samples: (i) firms in targeted sectors and those in control sectors not exposed to spillovers (i.e., partition set $\mathcal{T} \cup \mathcal{C}$, or $I_i=0$); and (ii) untreated firms (i.e., partition set $\mathcal{I} \cup \mathcal{C}$, or $G_i=0$). In the latter case, $G_i$ is replaced by $I_i$ in Equation S.11.

Figure \ref{fig:evstudy} summarizes the resulting estimates from Equation S.10 and S.11. Panel A-C show clear violations of the parallel trends assumption between the two regions, considering only the target sectors (panel A), control sectors (panel B), and control sectors exposed to spillovers (panel C). Instead, the second row of Figure \ref{fig:evstudy} reports tests on the parallel trend-in-trend assumptions. Panel D provides supportive evidence for the unconditional trend-in-trend assumption in Assumption \ref{ass::uptint} (the lead coefficient for 2014 has a $p$-value of 0.055). Hence, based pre-policy data is consistent with the conventional identifying assumptions used in TD models. However, this is not sufficient to justify interpreting the corresponding TD coefficients as the ATT when spillovers are present. Indeed, tests using only pre-policy data can indicate whether groups would have evolved similarly absent treatment, but in the presence of spillovers they cannot ensure that, after the policy, the comparison group still provides a valid proxy for the spillover-free counterfactual outcome \(Y(0,0)\).  Considering panels E and F, which test Assumptions \ref{ass:tinti0} and \ref{ass:tinti1}, we find no evidence of differential trend-in-trends, lending support to identification of the ATT and ASU under the DTD specification.

Table \ref{tab:SEZ_DTD_estimates} reports the main estimates of the TD and DTD models. The columns labelled TD and DTD estimate specifications without fixed effects. The three-way TD and three-way DTD columns report estimates from the saturated three-way fixed effects models, as defined in Sections S.10 of the Supplementary Material. Across the two TD specifications, the coefficient on the triple interaction term (\textit{T} $\times$ \textit{S} $\times$ \textit{G}) is stable at about 0.055--0.056 and statistically significant, implying a 5.5--5.6\% increase in employment for targeted-sector firms in the treated stratum under the no-spillovers interpretation. Once we allow for spillovers via DTD, the estimated ATT rises to about 6.6-6.8\%, and statistically significant. At the same time, the estimated spillover effect on exposed untreated sectors (\textit{T} $\times$ \textit{S} $\times$ \textit{I}) is positive and statistically significant (about 3.0-3.8\%). This pattern is consistent with positive spillovers attenuating standard TD estimates, in line with Proposition \ref{theo::triplediff_impact}. 
Note that the coefficient estimate of the triple interaction term in the TD model does not exactly correspond to the to the difference between its DTD counterpart and the DTD spillover estimate. This discrepancy arises because, in a multi-period setting with unbalanced control groups, pooled regression coefficients are not simple averages of year-by-year contrasts. The key qualitative implication is unchanged.

Finally, model fit diagnostics (AIC/BIC) favour the three-way DTD specification, which also delivers the most precise estimate of the spillover effect. We also conduct a robustness check in which we use distant municipalities in Campania (unlikely to be affected by the SEZ) as the control group, rather than Sicilian municipalities. The required trend-in-trend assumptions are also supported in this robustness exercise. Further details on municipality selection and the corresponding estimates are reported in Supplementary Material S.10. Overall, the robustness check confirms that the TD specification understates the employment effect of the SEZ, and it continues to detect spillover effects under the DTD specification. The estimated spillover coefficient is unchanged (3.8\%), while using non-neighbouring municipalities from the same region yields a larger estimated SEZ impact (9\%). Because these distant Campania municipalities are likely to differ systematically from those in Sicily (e.g.\ in proximity to ports and other infrastructure), we view the baseline estimate of 6.8\% as more credible; nevertheless, the combined evidence points to a statistically significant positive impact of the Campania SEZ on employment in the first three years after implementation.

\begin{table}[htbp]
\begin{threeparttable}
\centering
\caption{TD and DTD estimates}
	\label{tab:SEZ_DTD_estimates}
	\begin{tabular}{lcccc}
		\toprule
		& TD & three-way TD & DTD & three-way DTD \\ 
		\midrule
		T × G & -0.011 &  & -0.008 &  \\
		& (0.010) &  & (0.010) &  \\
		T × S & 0.021*** &  & 0.011$^+$ &  \\
		& (0.005) &  & (0.006) &  \\
		T × S × G & 0.056*** & 0.055*** & 0.066*** & 0.068*** \\
		& (0.013) & (0.013) & (0.014) & (0.013) \\
		T × I &  &  & 0.007 &  \\
		&  &  & (0.008) &  \\
		T × S × I &  &  & 0.030** & 0.038*** \\
		&  &  & (0.011) & (0.008) \\
		\midrule
		Num. Obs. & 454,985 & 454,985 & 454,985 & 454,985 \\
		AIC & 1,215,867.6 & 1,215,627.5 & 1,215,782.5 & 1,215,546.6 \\
		BIC & 5,270,760.1 & 5,270,387.8 & 5,270,653.1 & 5,270,295.8 \\
		\bottomrule
	\end{tabular}
    \begin{tablenotes}
    \scriptsize \item \textit{Notes:} The TD and DTD columns report OLS estimates from baseline post-dummy specifications. In the table, the indicator variable $S$ identifies firms in Campania SEZ municipalities, $G$ identifies firms in targeted sectors, $I$ identifies firms in non-targeted sectors exposed to spillovers, and $T$ identifies policy periods (2018-2020). The three-way columns report the corresponding saturated three-way fixed effects specifications with firm fixed effects and interactions between year fixed effects and the group indicators; consequently, lower-order post interactions are absorbed and not reported, whereas the yearly interactions are reported in Table S.1 in Supplementary Material. \\ 
		Standard errors are in parentheses. $^{+}p<0.1$, * $p<0.05$, ** $p<0.01$, *** $p<0.001$.
  \end{tablenotes}
\end{threeparttable}
\end{table}

\section{Conclusions}
\label{Sec:Conclusions}

\par This article shows that, in the presence of spillover effects, standard triple-difference (TD) designs generally do not identify the ATT, even when the relevant parallel trend-in-trends assumption is supported in the pre-policy period. The reason is that the conventional TD coefficient is, in general, a composite estimand that combines the treatment and spillover effects. For example, TD estimates of a positive ATT under a positive spillover are downward biased.

\par We address this identification problem by proposing a double-triple-difference (DTD) specification that separately identifies the ATT and the spillover effect. We make explicit the identification conditions underlying DTD, which rest on two distinct trend-in-trend assumptions that can be assessed in the pre-policy period. Consistency requires (i) spillovers to operate along only one of the two control dimensions (the treatment stratum), with the other dimension unaffected by spillovers, and (ii) the existence of a subset of control units within the treatment stratum that is not exposed to spillovers and can therefore serve as a clean comparison group.

\par We evaluate the finite-sample performance of TD and DTD estimators via Monte Carlo simulations, under different spillover scenarios. When spillovers are confined to the treatment stratum, DTD recovers unbiased estimates of both the ATT and ASU, whereas TD exhibits bias that increases with the prevalence of spillover exposure. When interference additionally crosses the treatment and placebo strata, both TD and DTD models are misspecified and can lead to biased estimates, underscoring the importance of carefully delineating the scope of interference in applied work.

\par We illustrate the empirical relevance of these results in an application to the introduction of a Special Economic Zone in Campania. The application discusses plausible spillover channels, implements pre-policy diagnostic tests aligned with the identifying assumptions, and compares estimates from TD and DTD specifications. TD designs are increasingly used in applied research, yet spillovers are often mentioned only briefly, if at all. Our framework provides a practical way to articulate the relevant estimands and assumptions when interference is a concern, and we encourage future work using TD designs to explicitly assess the plausibility of spillovers and, where appropriate, adopt estimators that remain valid under such interference.

\section*{Funding}
This work was partially funded by PNRR funds, PE10 project – ONFOODS, ‘Research and innovation network on food and nutrition. Sustainability, Safety and Security – Working ON Foods’ (code PE0000003, CUP J33C22002860001, 2022–2025). 

\bibliography{biblio.bib}
\bibliographystyle{apalike}

\section*{Supplementary Material}

\setcounter{subsection}{0}    
\setcounter{figure}{0}        
\setcounter{table}{0}         
\setcounter{equation}{0}      

\renewcommand{\thesubsection}{S\thesection.\arabic{subsection}}  
\renewcommand{\thefigure}{S\arabic{figure}}          
\renewcommand{\thetable}{S\arabic{table}}            
\renewcommand{\theequation}{S\arabic{equation}}     

\subsection{Proof of Proposition 1}
\begin{proof} 
We recall the definition of $\delta$ in the TD model of \textcolor{black}{Equation (1)} and translate to a specification that includes potential outcomes, by considering Assumptions \textcolor{black}{1, 2 and 3}, as follows
\begin{align*}
\delta = & \mathbb{E}[Y_{i1}(1,0) -Y_{i0}(0,0)\mid S = 1, G = 1] - \mathbb{E}[Y_{i1}(0,1)-Y_{i0}(0,0) \mid S = 1, G = 0]+\\
&- \Big( \mathbb{E}[Y_{i1}(0,0)-Y_{i0}(0,0) \mid S = 0, G = 1] - \mathbb{E}[Y_{i1}(0,0)-Y_{i0}(0,0) \mid S = 0, G = 0] \Big)\\
= & \mathbb{E}[Y_{i1}(1,0) -Y_{i1}(0,0)\mid S = 1, G = 1]  +  \mathbb{E}[Y_{i1}(0,0) -Y_{i0}(0,0)\mid S = 1, G = 1]+\\
&- \mathbb{E}[Y_{i1}(0,0)-Y_{i0}(0,0) \mid S = 1, G = 0]- \Big( \mathbb{E}[Y_{i1}(0,0)-Y_{i0}(0,0) \mid S = 0, G = 1] +\\
&- \mathbb{E}[Y_{i1}(0,0)-Y_{i0}(0,0) \mid S = 0, G = 0] \Big) - \mathbb{E}[Y_{i1}(0,1)-Y_{i1}(0,0) \mid S = 1, G = 0]\\
=& ATT(S=1,G=1) + \Delta^{TT} - ASU(S = 1, G = 0).
\end{align*}
Note that, under \textcolor{black}{Assumption 4}, $ \Delta^{TT}=0$, and therefore, the final result aligns with the statement of the proposition.
\end{proof}

\subsection{Proof of Proposition 2}
\begin{proof} 
We explicitly define $\psi$ under TD model as 
\begin{align*}
\psi = & \mathbb{E}[Y_{i1} -Y_{i0}\mid S = 1, G = 0] - \mathbb{E}[Y_{i1}-Y_{i0} \mid S = 0, G = 0],
\end{align*}
that, by exploiting the potential outcome framework and considering Assumptions \textcolor{black}{1, and 2}, can be rewritten as
\begin{align*}
\psi = & \mathbb{E}[Y_{i1} (0,1) -Y_{i0}(0,0)\mid S = 1, G = 0] - \mathbb{E}[Y_{i1}(0,0)-Y_{i0}(0,0) \mid S = 0, G = 0]\\
=& \mathbb{E}[Y_{i1} (0,1) -Y_{i1}(0,0)\mid S = 1, G = 0] + \mathbb{E}[Y_{i1} (0,0) -Y_{i0}(0,0)\mid S = 0, G = 0] + \\&- \mathbb{E}[Y_{i1}(0,0)-Y_{i0}(0,0) \mid S = 0, G = 0]\\
=& ASU(S = 1, G = 0) + \Delta^{T}_{G=0}.
\end{align*}
Note that, under \textcolor{black}{Assumption 5}, $\Delta^{T}_{G=0}=0$, and therefore, the final result aligns with the statement of the proposition. 
\end{proof}

\subsection{Identifying assumption in case of multiple spillovers}

To further investigate the identification of the triple difference parameter in presence of spillovers, we consider the assumption of stratum SUTVA (Assumption 1) as just one possible case. Suppose spillover effects also influence the placebo stratum, leading to multiple simultaneous spillovers across different groups. Specifically, we analyze a case with two types of spillover effects: one affecting the control group in the treatment stratum $\mathcal{I}_1$, as outlined in Assumption 1, and another influencing units within a specific group of the placebo stratum, with two possible cases.
In particular, Case A examines the case where the group impacted by spillover effects in the placebo stratum consists of potential beneficiaries, namely $\mathcal{T}_0$. This case is illustrated in Figure \ref{fig::scenario1}.

\begin{proposition} Under Assumptions 2, 3, 4 and Case A, the triple interaction parameter $\delta$ is defined as follows 
\begin{align*}
\delta = ATT(S=1,G=1)- ASU(S = 0, G = 1)-ASU(S = 1, G = 0).
\end{align*}
Thus, it identifies the ATT if and only if all ASU are equal to zero.
\end{proposition}
\begin{proof} 
We recall the definition of $\delta$ in the TD model of \textcolor{black}{Equation (1)} and translate to a specification that includes potential outcomes, by considering Assumptions \textcolor{black}{2 and 3} and Case A, as follows
\begin{align*}
\delta = & \mathbb{E}[Y_{i1}(1,0) -Y_{i0}(0,0)\mid S = 1, G = 1] - \mathbb{E}[Y_{i1}(0,1)-Y_{i0}(0,0) \mid S = 1, G = 0]+\\
&- \Big( \mathbb{E}[Y_{i1}(0,1)-Y_{i0}(0,0) \mid S = 0, G = 1] - \mathbb{E}[Y_{i1}(0,0)-Y_{i0}(0,0) \mid S = 0, G = 0] \Big)\\
= & \mathbb{E}[Y_{i1}(1,0) -Y_{i1}(0,0)\mid S = 1, G = 1]  +  \mathbb{E}[Y_{i1}(0,0) -Y_{i0}(0,0)\mid S = 1, G = 1]+\\
&- \mathbb{E}[Y_{i1}(0,0)-Y_{i0}(0,0) \mid S = 1, G = 0]- \Big( \mathbb{E}[Y_{i1}(0,0)-Y_{i0}(0,0) \mid S = 0, G = 1] +\\
&- \mathbb{E}[Y_{i1}(0,0)-Y_{i0}(0,0) \mid S = 0, G = 0] \Big) - \mathbb{E}[Y_{i1}(0,1)-Y_{i1}(0,0) \mid S = 0, G = 1]+
\\ &- \mathbb{E}[Y_{i1}(0,1)-Y_{i1}(0,0) \mid S = 1, G = 0]\\
=&  ATT(S=1,G=1) + \Delta^{TT} - ASU(S = 0, G = 1)-ASU(S = 1, G = 0).
\end{align*}
Note that, under \textcolor{black}{Assumption 4}, $ \Delta^{TT}=0$, and therefore, the final result aligns with the statement of the proposition.
\end{proof}
\begin{figure} 
\begin{tikzpicture}[scale=1.3, transform shape]
\node[set,text width=3cm,pattern=north west lines, pattern color=gray] (leftbg) at (0,-0.4) {};
\node[set,preaction={fill=white},pattern=dots,text width=1.5cm] (left) at (0,0) {};
\node[above=0.7cm of left] {\scriptsize Treated stratum $\mathcal{S}_1$};
\node at (0,0) {\scriptsize$\mathcal{T}_1$}; 
\node at (0,-1.2) {\scriptsize$\mathcal{I}_1$}; 

\node[set, text width=3cm] (rightbg) at (4,-0.4) {};
\node[set,text width=1.5cm, pattern=north west lines, pattern color=gray] (right) at (4,0) {};
\node[above=0.7cm of right] {\scriptsize Placebo stratum $\mathcal{S}_0$};
\node at (4,0) {\scriptsize$\mathcal{T}_0$}; 
\node at (4,-1.2) {\scriptsize$\mathcal{I}_0$}; 
\matrix[draw=white, matrix of nodes, column sep=1mm]
          (m) at (8,0) {
  Legend \\
  \node [set,text width=3mm, pattern=dots, pattern color=gray, draw=black]{}; & Treated units set\\
   \node [set,text width=3mm, pattern=north west lines, pattern color=gray, draw=black]{}; & Spillover units set\\
\node [set,text width=3mm, draw=black]{}; & Unaffected units set\\
 };
\end{tikzpicture}
\caption{Multiple spillovers, Case A.}
\label{fig::scenario1}
\end{figure}
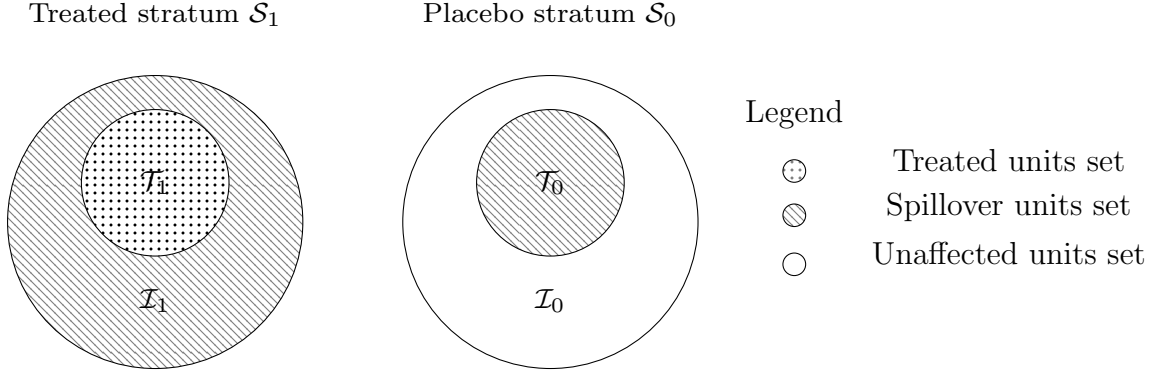

Moreover, Case B examines the case where the spillover effect in the placebo stratum impacts the group of potential non-beneficiaries, specifically $\mathcal{I}_0$. This case is illustrated in Figure \ref{fig::scenario2}. 

\begin{proposition} Under assumptions 2, 3, 4, and Case B, the triple interaction parameter $\delta$ is defined as follows 
\begin{align*}
\delta =  ATT(S=1,G=1) -ASU(S = 1, G = 0) - ASU(S = 0, G = 0).
\end{align*}
Thus, it identifies the ATT if and only if all ASU are equal to zero.
\end{proposition}
\begin{proof} 
We recall the definition of $\delta$ in the TD model of \textcolor{black}{Equation (1)} and translate to a specification that includes potential outcomes, by considering Assumptions \textcolor{black}{2 and 3} and Case B, as follows
\begin{align*}
\delta = & \mathbb{E}[Y_{i1}(1,0) -Y_{i0}(0,0)\mid S = 1, G = 1] - \mathbb{E}[Y_{i1}(0,1)-Y_{i0}(0,0) \mid S = 1, G = 0]+\\
&- \Big( \mathbb{E}[Y_{i1}(0,0)-Y_{i0}(0,0) \mid S = 0, G = 1] - \mathbb{E}[Y_{i1}(0,1)-Y_{i0}(0,0) \mid S = 0, G = 0] \Big)\\
= & \mathbb{E}[Y_{i1}(1,0) -Y_{i1}(0,0)\mid S = 1, G = 1]  +  \mathbb{E}[Y_{i1}(0,0) -Y_{i0}(0,0)\mid S = 1, G = 1]+\\
&- \mathbb{E}[Y_{i1}(0,0)-Y_{i0}(0,0) \mid S = 1, G = 0]- \Big( \mathbb{E}[Y_{i1}(0,0)-Y_{i0}(0,0) \mid S = 0, G = 1] +\\
&- \mathbb{E}[Y_{i1}(0,0)-Y_{i0}(0,0) \mid S = 0, G = 0] \Big) +\mathbb{E}[Y_{i1}(0,1)-Y_{i1}(0,0) \mid S = 0, G = 0]+
\\ &- \mathbb{E}[Y_{i1}(0,1)-Y_{i1}(0,0) \mid S = 1, G = 0]\\
=&  ATT(S=1,G=1) + \Delta^{TT} -ASU(S = 1, G = 0) + ASU(S = 0, G = 0).
\end{align*}
Note that, under \textcolor{black}{Assumption 4}, $ \Delta^{TT}=0$, and therefore, the final result aligns with the statement of the proposition.
\end{proof}

\begin{figure} 
\begin{tikzpicture}[scale=1.3, transform shape]
\node[set,text width=3cm,pattern=north west lines, pattern color=gray] (leftbg) at (0,-0.4) {};
\node[set,preaction={fill=white},pattern=dots,text width=1.5cm] (left) at (0,0) {};
\node[above=0.7cm of left] {\scriptsize Treatment stratum $\mathcal{S}_1$};
\node at (0,0) {\scriptsize$\mathcal{T}_1$}; 
\node at (0,-1.2) {\scriptsize$\mathcal{I}_1$}; 

\node[set,text width=3cm, pattern=north west lines, pattern color=gray] (rightbg) at (4,-0.4) {};
\node[set,preaction={fill=white},text width=1.5cm] (right) at (4,0) {};
\node[above=0.7cm of right] {\scriptsize Placebo stratum $\mathcal{S}_0$};
\node at (4,0) {\scriptsize$\mathcal{T}_0$}; 
\node at (4,-1.2) {\scriptsize$\mathcal{I}_0$}; 
\matrix[draw=white, matrix of nodes, column sep=1mm]
          (m) at (8,0) {
  Legend \\
  \node [set,text width=3mm, pattern=dots, pattern color=gray, draw=black]{}; & Treated units set\\
   \node [set,text width=3mm, pattern=north west lines, pattern color=gray, draw=black]{}; & Spillover units set\\
\node [set,text width=3mm, draw=black]{}; & Unaffected units set\\
 };
\end{tikzpicture}
\caption{Multiple spillovers, Case B.}
\label{fig::scenario2}
\end{figure}
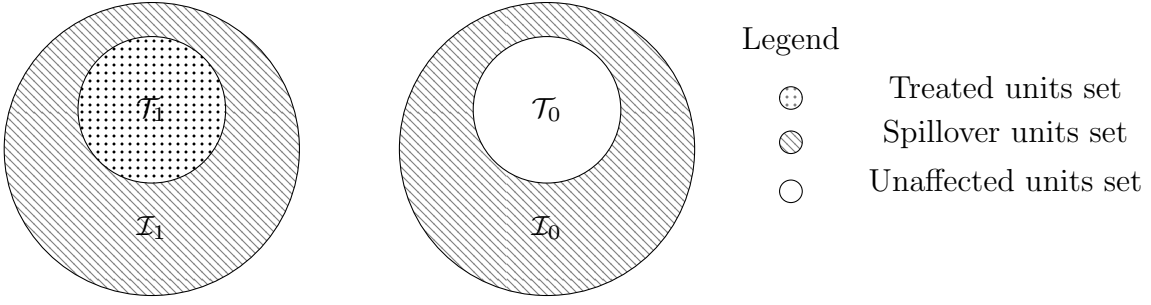
Note that in both Cases A and B, the TD parameter is not well identified. The non-identifying biases caused by multiple spillovers overlap and cancel out only when the two spillovers have, on average, the same magnitude but opposite signs (Case A) or same magnitude and same sign (Case B). Therefore, this situation arises when the treatment generates a diffusion effect in one group and a displacement effect in the other, and when these opposing effects are of equal magnitude (Case A), or when both groups experience diffusion or displacement effects of equal magnitude (Case B).

\subsection{Proof of Proposition 3}

\begin{proof}
We recall the definition of $\delta$ in the DTD model of \textcolor{black}{Equation (2)} and translate to a specification that includes potential outcomes, by considering Assumptions \textcolor{black}{1, 2, 3 and 6}, as follows
\begin{align*}
\delta = & \mathbb{E}[Y_{i1}(1,0) -Y_{i0}(0,0)\mid S = 1, G = 1, I=0] - \mathbb{E}[Y_{i1}(0,0)-Y_{i0}(0,0) \mid S = 1, G = 0, I=0]+\\
&- \Big( \mathbb{E}[Y_{i1}(0,0)-Y_{i0}(0,0) \mid S = 0, G = 1, I=0] - \mathbb{E}[Y_{i1}(0,0)-Y_{i0}(0,0) \mid S = 0, G = 0, I=0] \Big)\\
= & \mathbb{E}[Y_{i1}(1,0) -Y_{i1}(0,0)\mid S = 1, G = 1, I=0]  +  \mathbb{E}[Y_{i1}(0,0) -Y_{i0}(0,0)\mid S = 1, G = 1, I=0]+\\
&- \mathbb{E}[Y_{i1}(0,0)-Y_{i0}(0,0) \mid S = 1, G = 0, I=0]- \Big( \mathbb{E}[Y_{i1}(0,0)-Y_{i0}(0,0) \mid S = 0, G = 1, I=0] +\\
&- \mathbb{E}[Y_{i1}(0,0)-Y_{i0}(0,0) \mid S = 0, G = 0, I=0] \Big) \\
=&  ATT(S=1,G=1, I=0) + \Delta_{I=0}^{TT}.
\end{align*}
Note that, under \textcolor{black}{Assumption 7}, $ \Delta_{I=0}^{TT}=0$, and therefore, the final result aligns with the statement of the proposition.
\end{proof}

\subsection{Proof of Proposition 4}

\begin{proof}
We recall the definition of $\psi$ in the DTD model in \textcolor{black}{Equation (2)} and translate to a specification that includes potential outcomes by considering Assumptions \textcolor{black}{1, 2 and 6}, as follows
\begin{align*}
\psi = & \mathbb{E}[Y_{i1}(1,0) -Y_{i0}(0,0)\mid S = 1, G = 0, I=1] - \mathbb{E}[Y_{i1}(0,0)-Y_{i0}(0,0) \mid S = 1, G = 0, I=0]+\\
&- \Big( \mathbb{E}[Y_{i1}(0,0)-Y_{i0}(0,0) \mid S = 0, G = 0, I=1] - \mathbb{E}[Y_{i1}(0,0)-Y_{i0}(0,0) \mid S = 0, G = 0, I=0] \Big)\\
= & \mathbb{E}[Y_{i1}(1,0) -Y_{i1}(0,0)\mid S = 1, G = 0, I=1]  +  \mathbb{E}[Y_{i1}(0,0) -Y_{i0}(0,0)\mid S = 1, G = 0, I=1]+\\
&- \mathbb{E}[Y_{i1}(0,0)-Y_{i0}(0,0) \mid S = 1, G = 0, I=0]- \Big( \mathbb{E}[Y_{i1}(0,0)-Y_{i0}(0,0) \mid S = 0, G = 0, I=1] +\\
&- \mathbb{E}[Y_{i1}(0,0)-Y_{i0}(0,0) \mid S = 0, G = 0, I=0] \Big) \\
=&  ASU(S=1, G=0, I=1) + \Delta_{G=0}^{TT}.
\end{align*}
Note that, under \textcolor{black}{Assumption 8}, $ \Delta_{G=0}^{TT}=0$, and therefore, the final result aligns with the statement of the proposition.
\end{proof}

\subsection{Auxiliary Lemmas}

\begin{lemma}[] \label{lemma::w_01}
For a generic stratum $s'$, groups $g'$ or $j'$, we can state that  
\begin{align}\label{eq::lemma11}
\mathbb{E}[w_{is',g'}^{S=1,G=1,I=0}(X)(Y_{i1}-Y_{i0})] = \mathbb{E}[\mathbb{E}[Y_{i1}-Y_{i0}| X, S=s', G=g',I=0]| S=1, G=1, I=0],\\
\mathbb{E}[w_{s',j'}^{S=1,G=0,I=1}(X)(Y_{i1}-Y_{i0})] = \mathbb{E}[\mathbb{E}[Y_{i1}-Y_{i0}| X, S=s', G=0,I=j']| S=1, G=0, I=1].\label{eq::lemma12}
\end{align}
\end{lemma}

\begin{proof} We prove result in Equation \eqref{eq::lemma11} as result in Equation \eqref{eq::lemma12} follows directly. Let us define the quantity by exploiting the definition of $w_{is',g'}^{S=1,G=1,I=0}(X)$ as
\begin{align}
\mathbb{E}[w_{is',g'}^{S=1,G=1,I=0}(X)(Y_{i1}-Y_{i0})] = 
\frac{\mathbb{E}\bigg[\mathbf{1}_{i}\{S=s, G=g,I=0\} \cdot \frac{p^{S=1,G=1,I=0}_{s,g}(X)}{(1-p^{S=1,G=1,I=0}_{s,g}(X))}(Y_{i1}-Y_{i0})\bigg]}
{\mathbb{E} \left[ \mathbf{1}_{i}\{S=s, G=g,I=0\} \frac{p^{S=1,G=1,I=0}_{s,g}(X)}{(1-p^{S=1,G=1,I=0}_{s,g}(X))} \right]}.
\label{eq::1}
\end{align}
By considering the law of iterated expectation and conditional probabilities, we can reexpress the enumerator of Equation \eqref{eq::1} as 
\begin{align} \nonumber
&\mathbb{E}\bigg[\mathbf{1}_{i}\{S=s, G=g,I=0\} \cdot \frac{p^{S=1,G=1,I=0}_{s,g}(X)}{(1-p^{S=1,G=1,I=0}_{s,g}(X))}(Y_{i1}-Y_{i0})\bigg] = \\\nonumber
&=\mathbb{E}\bigg[\frac{\mathbb{E}[ \mathbf{1}_{i} \{ S=1, G=1,I=0\} |X]}{\mathbb{E}[ \mathbf{1}_{i} \{ S=s, G=g,I=0\} |X]}\mathbb{E}[\mathbf{1}_{i}\{S=s, G=g,I=0\} \cdot(Y_{i1}-Y_{i0})|X]\bigg]\\\nonumber
&=\mathbb{E}\bigg[\mathbb{E}[ \mathbf{1}_{i} \{ S=1, G=1,I=0\} |X]\mathbb{E}[Y_{i1}-Y_{i0}|X, S=s, G=g, I=0]\bigg]\\
&=\mathbb{E}\bigg[\mathbf{1}_{i} \{ S=1, G=1,I=0\} \mathbb{E}[Y_{i1}-Y_{i0}|X, S=s, G=g, I=0]\bigg],
\label{eq::2}
\end{align}
and the denominator of Equation \eqref{eq::1} similarly as
\begin{align}
&\mathbb{E} \left[ \mathbf{1}_{i}\{S=s, G=g,I=0\} \frac{p^{S=1,G=1,I=0}_{s,g}(X)}{(1-p^{S=1,G=1,I=0}_{s,g}(X))} \right]=\nonumber\\
& =\mathbb{E}\bigg[\frac{\mathbb{E}[ \mathbf{1}_{i} \{ S=1, G=1,I=0\} |X]}{\mathbb{E}[ \mathbf{1}_{i} \{ S=s, G=g,I=0\} |X]}\mathbb{E}[\mathbf{1}_{i}\{S=s, G=g,I=0\} |X]\bigg]\nonumber\\
&=\mathbb{E}[ \mathbf{1}_{i} \{ S=1, G=1,I=0\}].
\label{eq::3}
\end{align}
Combining equations \eqref{eq::2} and \eqref{eq::3} with \eqref{eq::1} immediately proves the result.
\end{proof}

\begin{lemma}[Definition of ATT and ASU as an inverse probability weighting parameter]
Let us define the ATT as
\begin{align}\label{eq::lemma21}
ATT^{S=1,G=1,I=0}=& \mathbb{E}\bigg[w_i^{S=1,G=1,I=0} [Y_{i1}(1,0) -Y_{i0}(0,0)]\bigg] +\\&- \mathbb{E}\bigg[[w^{S=1,G=1,I=0}_{i1,0}(X)+w^{S=1,G=1,I=0}_{i0,1}(X)-w^{S=1,G=1,I=0}_{i0,0}(X)][Y_{i1}(0,0)-Y_{i0}(0,0)]\bigg]\nonumber
\end{align}
with $w_{s,g}$ for a generic stratum $s$, group $g$. Let us define the ASU as
\begin{align}
ASU^{S=1,G=0,I=1}=& \mathbb{E}\bigg[w_i^{S=1,G=0,I=1} [Y_{i1}(0,1) -Y_{i0}(0,0)]\bigg]+ \label{eq::lemma22}\\&- \mathbb{E}\bigg[[w^{S=1,G=0,I=1}_{i1,0}(X)+w^{S=1,G=0,I=1}_{i0,1}(X)-w^{S=1,G=0,I=1}_{i0,0}(X)][Y_{i1}(0,0)-Y_{i0}(0,0)]\bigg]\nonumber
\end{align}
with $w_{s,i}$ for a generic stratum $s$, group $i$. 
\label{lemma::w_02}
\end{lemma}
\begin{proof}
We prove the result in Equation \eqref{eq::lemma21}, as the result in Equation \eqref{eq::lemma22} follows analogously.
By the law of iterative expectation, let us define the $ATT^{\tiny{S=1,G=1,I=0}}$ as
\begin{align}
ATT^{S=1,G=1,I=0}&=\mathbb{E}[ATT^{S=1,G=1,I=0,X} |S=1,G=1,I=0] \\
&=\mathbb{E}[\mathbb{E}[Y_{i1}(1,0)-Y_{i1}(0,0)|S=1,G=1,I=0,X] |S=1,G=1,I=0]. \nonumber
\end{align}
Under \textcolor{black}{Assumption 9}, we can easily reexpress $ATT^{\tiny{S=1,G=1,I=0}}$  as 
\begin{align*}\nonumber
 =&\mathbb{E}[\mathbb{E}[Y_{i1}(1,0)-Y_{i1}(0,0)|S=1,G=1,I=0,X]+\Delta_{i \mid I=0,X}^{TT}|S=1,G=1,I=0]
\\
=& \mathbb{E}\bigg[  \mathbb{E}[Y_{i1}(1,0) -Y_{i0}(0,0)\mid S = 1, G = 1, I=0, X] - \mathbb{E}[Y_{i1}(0,0)-Y_{i0}(0,0) \mid S = 1, G = 0, I=0, X]+\nonumber\\
&-  \mathbb{E}[Y_{i1}(0,0)-Y_{i0}(0,0) \mid S = 0, G = 1, I=0, X] + \mathbb{E}[Y_{i1}(0,0)-Y_{i0}(0,0) \mid S = 0, G = 0, I=0, X]\\
&\bigg|S=1,G=1,I=0\bigg]\\
=& \mathbb{E}\bigg[  \mathbb{E}[Y_{i1}(1,0) -Y_{i0}(0,0)\mid S = 1, G = 1, I=0, X]\bigg|S=1,G=1,I=0\bigg]+ \\
&- \mathbb{E}\bigg[(w^{S=1,G=1,I=0}_{i1,0}(X)+w^{S=1,G=1,I=0}_{i0,1}(X)-w^{S=1,G=1,I=0}_{i0,0}(X))(Y_{i1}(0,0)-Y_{i0}(0,0))\bigg]\\
=& \mathbb{E}\bigg[ Y_{i1}(1,0) -Y_{i0}(0,0)\bigg|S=1,G=1,I=0\bigg] +\\&- \mathbb{E}\bigg[(w^{S=1,G=1,I=0}_{i1,0}(X)+w^{S=1,G=1,I=0}_{i0,1}(X)-w^{S=1,G=1,I=0}_{i0,0}(X))(Y_{i1}(0,0)-Y_{i0}(0,0))\bigg]\\
=& \mathbb{E}\bigg[w_i^{S=1,G=1,I=0} (Y_{i1}(1,0) -Y_{i0}(0,0))\bigg] +\\&- \mathbb{E}\bigg[(w^{S=1,G=1,I=0}_{i1,0}(X)+w^{S=1,G=1,I=0}_{i0,1}(X)-w^{S=1,G=1,I=0}_{i0,0}(X))(Y_{i1}(0,0)-Y_{i0}(0,0))\bigg].
\end{align*}
The second equality is simple algebra, the third equality is given by Lemma \ref{lemma::w_01}, the fourth equality is given by the law of iterative expectation, and the fifth equality follows from the definition of $w_i^{S=1, G=1, I=0}$.
\end{proof}

\subsection{Proof of Proposition 5}
\begin{proof}
Let us rexpress though simple algebra the Equation \textcolor{black}{(4)} as
\begin{align*}
\delta_{\texttt{dr}} &= \mathbb{E} \bigg[ 
\Big( w^{S=1,G=1,I=0}_{i} - w^{S=1,G=1,I=0}_{i1,1}(X)  - w^{S=1,G=1,I=0}_{i0,1}(X) +w^{S=1,G=1,I=0}_{i0,0}(X) \Big) 
( Y_{i1} - Y_{i0} )  \bigg]+ \nonumber \\
&-\mathbb{E} \bigg[ 
\Big( w^{S=1,G=1,I=0}_{i} - w^{S=1,G=1,I=0}_{i1,1}(X) \Big) 
m_i^{S=1,G=0,I=0}(X) \bigg]+ \nonumber \\
& - \mathbb{E} \bigg[ 
\Big( w^{S=1,G=1,I=0}_{i} - w^{S=1,G=1,I=0}_{i0,1}(X) \Big) 
m_i^{S=0,G=1,I=0}(X) \bigg]+ \nonumber \\
& + \mathbb{E} \bigg[ 
\Big( w^{S=1,G=1,I=0}_{i} - w^{S=1,G=1,I=0}_{i0,0}(X) \Big) 
m_i^{S=0,G=0,I=0}(X) \bigg], \quad \forall i.
\end{align*}
Let us expand it to incorporate the potential outcome framework as
\begin{align*}
\delta_{\texttt{dr}} &= \mathbb{E} \bigg[ 
 w^{S=1,G=1,I=0}_{i}  
[ Y_{i1}(0,1) - Y_{i0}(0,0)] \bigg]+\\
&-\mathbb{E} \bigg[ 
\Big( w^{S=1,G=1,I=0}_{i1,1}(X) + w^{S=1,G=1,I=0}_{i0,1}(X) -w^{S=1,G=1,I=0}_{i0,0}(X) \Big) 
[Y_{i1}(0,0) - Y_{i0}(0,0) ]  \bigg]+ \nonumber \\
&-\mathbb{E} \bigg[ 
\Big( w^{S=1,G=1,I=0}_{i} - w^{S=1,G=1,I=0}_{i1,1}(X) \Big) 
m_i^{S=1,G=0,I=0}(X) \bigg]+ \nonumber \\
& - \mathbb{E} \bigg[ 
\Big( w^{S=1,G=1,I=0}_{i} - w^{S=1,G=1,I=0}_{i0,1}(X) \Big) 
m_i^{S=0,G=1,I=0}(X) \bigg]+ \nonumber \\
& + \mathbb{E} \bigg[ 
\Big( w^{S=1,G=1,I=0}_{i} - w^{S=1,G=1,I=0}_{i0,0}(X) \Big) 
m_i^{S=0,G=0,I=0}(X) \bigg].
\end{align*}
By applying Lemma \ref{lemma::w_02}, we can state that
\begin{align*}
\delta_{\texttt{dr}} &=ATT^{S=1,G=1,I=0}
-\mathbb{E} \bigg[ 
\Big( w^{S=1,G=1,I=0}_{i} - w^{S=1,G=1,I=0}_{i1,1}(X) \Big) 
m_i^{S=1,G=0,I=0}(X) \bigg]+ \nonumber \\
& - \mathbb{E} \bigg[ 
\Big( w^{S=1,G=1,I=0}_{i} - w^{S=1,G=1,I=0}_{i0,1}(X) \Big) 
m_i^{S=0,G=1,I=0}(X) \bigg] +\nonumber \\
& + \mathbb{E} \bigg[ 
\Big( w^{S=1,G=1,I=0}_{i} - w^{S=1,G=1,I=0}_{i0,0}(X) \Big) 
m_i^{S=0,G=0,I=0}(X) \bigg] =ATT^{S=1,G=1,I=0}.
\end{align*}
The latter equality is due to $w^{S=1,G=1,I=0}_{i} - w^{S=1,G=1,I=0}_{i1,1}(X)= w^{S=1,G=1,I=0}_{i} - w^{S=1,G=1,I=0}_{i0,1}(X)=w^{S=1,G=1,I=0}_{i} - w^{S=1,G=1,I=0}_{i0,0}(X)=0$, as it can be easily proven through the law of iterated expectation.
\end{proof}
\subsection{Proof of Proposition 6}
\begin{proof}
The proof proceeds analogously to that of Proposition~5, upon replacing 
$G = g$ (for all $g$) with $G = 0$, 
and $I = 0$ with $I = j$, 
redefining the weights as $w^{S=1,G=0,I=1}_{isj}(X)$ instead of $w^{S=1,G=1,I=0}_{isg}(X)$, considering \textcolor{black}{Assumption 10 rather than Assumption 9},
and substituting the potential outcomes $Y_{i1}(1,0)$ with $Y_{i1}(0,1)$ for the generic $i$.
\end{proof}

\subsection{Details about the simulation study of Subsection 5.2}

Time-invariant covariates are generated in line with \citet{kang2007demystifying, sant2020doubly} as follows. Firstly, for each unit $i$, $Z_i =(Z_{i1},Z_{i2},Z_{i3},Z_{i4})'$ is generated by a $\mathcal{N}(0, I_4)$ with $I_4$, a $4\times 4$ identity matrix. Then, a transformation is operated as follows $\tilde{X}_i=(\exp(0.5Z_{i1}),10+Z_{i2}/(1+\exp(Z_{i1})), (0.6+Z_{i1}Z_{i3}/25)^3, (20+Z_{i1}+Z_{i4})^2)'$. The resulting $X_i$ is the standardized version of $\tilde{X}_i$. 
The predictor are then defined in line with \citet{ortiz2025better} as follows:
\begin{align*}
    f^{ps}_{S=s, G=g}(X_i) &= (0.05 S_i + 0.2 (1-S_i)) X_i'\gamma_{s,g},\\
     f^{reg}_{S=s, G=g}(X_i) &= 2010 + S_i X_i' \beta_1 +  (1-S_i) X_i' \beta_0,
\end{align*}
 with $\gamma_{0,0}=(-1, 0.5, -0.25, -0.1)'$, $\gamma_{0,1}=(-0.5, 2, 0.5, -0.2)'$, $\gamma_{i1,0}=(3, -1.5, 0.75, -0.3)'$, $\beta_1 = (27.4,13.7,13.7,13.7)'$ and $\beta_0 = 0.5 \beta_1$.  Note that $f^{ps}_{S=1, G=1}(X_i)=1$. 
 Lastly, $\nu_i(X_i, S_i, G_i) \sim \mathcal{N}(2010 G_i + S_i G_i X_i' \beta_1 +  (1-S_i) G_i X_i' \beta_0,1)$. 
See \citet{ortiz2025better} for additional details.  Consider U to be a uniform random variable in $[0,1]$, the assignment process of units to sub-groups is as follows
\begin{align}
    \begin{cases}
        S=0 \quad \text{and} \quad G=0 \quad \text{if} \quad U \leq p^{S=0,P=0}(X)\\
          S=0 \quad \text{and}\quad  G=1  \quad \text{if} \quad p^{S=0,P=0}(X) < U \leq \sum_{j \in \lbrace 0,1 \rbrace} p^{S=0, P=j}(X)\\
          S=1 \quad \text{and}\quad  G=0   \quad \text{if} \quad \sum_{j \in \lbrace 0,1 \rbrace} p^{S=0,P=J}(X) < U \leq \sum_{j \in \lbrace 0,1 \rbrace} p^{S=0, P=j}(X) + p^{S=1, P=0}(X)\\
            S=1 \quad \text{and}\quad  G=1 \quad \text{if} \quad U >\sum_{j \in \lbrace 0,1 \rbrace} p^{S=0, P=j}(X) + p^{S=1, P=0}(X).\\
    \end{cases}
\end{align}

\subsection{Additional details on the application of Section 6: SEZ in Campania}

\subsubsection*{Event study specifications to test parallel trends and parallel trend-in-trends assumptions}

Let $L_{t=j}$ with $j=2014,\dots,2017$ denote a lead indicator variable which is one when $t=j$ and 0 otherwise, with 2013 serving as the omitted base year. Any departure from parallel trend in a given year is captured by the coefficients $\delta_j$. 
As a starting point, let us consider the DiD-style parallel trends requirement. This can be probed by estimating the following equation in the pre-policy period:
\begin{equation}
	\label{eq:tt_DiD}
	\begin{split}
		Y_{it}
		= \beta_i + \beta_t + \sum_{j=2014}^{2017} \delta_j\big(S_i \times L_{t=j}\big) + \varepsilon_{it}, \quad t=2013,\dots,2017. \\
	\end{split}
\end{equation}
where $\beta_i$ and $\beta_t$ denote firm and year fixed effects. Under parallel trends, the pre-event coefficients $\delta_j$ should be jointly equal to zero. We estimate the model in Equation \eqref{eq:tt_DiD} on the sub-samples induced by the partition set (i) $\mathcal{T}$ (targeted sectors); (ii) $\mathcal{I} \cup \mathcal{C}$ (all non-targeted sectors); (iii) $\mathcal{I}$ (non-targeted sectors exposed to spillover effects).  
Analogously, parallel trend-in-trends assumption can be examined by interacting the leads with both partitions of the TD design:
\begin{equation}
	\label{eq:es_tt_baseline}
	\begin{split}
		Y_{it}
		= &\beta_i + \beta_t + \sum_{j=2014}^{2017} \gamma_j\big(S_i \times L_{t=j}\big) + \sum_{j=2014}^{2017} \eta_j\big(G_i \times L_{t=j}\big)
		+ \sum_{j=2014}^{2017} \delta_j\big(S_i \times G_i \times L_{t=j}\big)	+ \varepsilon_{it}, 
	\end{split}
\end{equation}
for $t=2013,\dots,2017$.
In this specification, the coefficients of interest are the lead coefficients on $S_i \times G_i \times L_{t=j}$: failing to reject $\delta_j=0$ for all $j$ provides evidence consistent with parallel trend-in-trends assumptions.

\subsubsection*{Three-way fixed effect specifications of TD and DTD models and additional results}

Three-way fixed effect specifications as implemented in the empirical application for the TD model:
\begin{align}
	Y_{it} \;=\;& \beta_i + \beta_t + \beta_{t,S} \times S_i +  \beta_{t,G} \times G_i  + \delta (S_i \times G_i \times T_{t}) + \epsilon_{it},
\end{align}
with $\beta_{t,G}$ and $\beta_{t,S}$ group-specific time fixed effects;
and for the DTD model:
\begin{align}
Y_{it} \;=\;& \beta_i + \beta_t + \beta_{t,S} \times S_i +  \beta_{t,G} \times G_i  + \delta (S_i \times G_i \times T_{t}) 
	+ \psi (S_i\times I_i\times T_{t})
	+ \epsilon_{it}.
\end{align}
Table \ref{tab:triple_results} integrates the estimates shown in Table 4 of the paper, by providing the coefficients of yearly interactions with the indicator variables for the treatment  and target group partitions, using 2020 as the baseline year. 
\begin{table}[hb]
\centering
\begin{threeparttable}
\caption{Three-way fixed effect TD and DTD estimates -- yearly interactions coefficients.}
\label{tab:triple_results}
\begin{tabular}{lcccc}
\toprule
 & \multicolumn{2}{c}{\textbf{Triple}} & \multicolumn{2}{c}{\textbf{Double-triple}} \\
\cmidrule(lr){2-3} \cmidrule(lr){4-5}
 & Coef. & S.E. & Coef. & S.E. \\
\midrule
S $\times$ 2013 & -0.066*** & (0.009) & -0.052*** & (0.009) \\
S $\times$ 2014 & -0.053*** & (0.008) & -0.040*** & (0.008) \\
S $\times$ 2015 & -0.022**  & (0.007) & -0.009    & (0.007) \\
S $\times$ 2016 & -0.016*   & (0.007) & -0.003    & (0.007) \\
S $\times$ 2017 & -0.004    & (0.006) & 0.009     & (0.006) \\
S $\times$ 2018 & -0.014**  & (0.005) & -0.014**  & (0.005) \\
S $\times$ 2019 & -0.002    & (0.004) & -0.002    & (0.004) \\
G $\times$ 2013 & -0.014    & (0.014) & -0.014    & (0.014) \\
G $\times$ 2014 & -0.013    & (0.013) & -0.013    & (0.013) \\
G $\times$ 2015 & -0.004    & (0.012) & -0.004    & (0.012) \\
G $\times$ 2016 & 0.000     & (0.012) & 0.000     & (0.012) \\
G $\times$ 2017 & -0.008    & (0.011) & -0.008    & (0.011) \\
G $\times$ 2018 & -0.030*** & (0.007) & -0.030*** & (0.007) \\
G $\times$ 2019 & -0.021*** & (0.005) & -0.021*** & (0.005) \\
\bottomrule
\end{tabular}

\begin{tablenotes}
\scriptsize
\item \textit{Notes:} The table reports OLS estimates of yearly interaction coefficients from three-way fixed effects TD and DTD models. $S$ identifies firms in Campania SEZ municipalities, $G$ identifies firms in targeted sectors.
\item Standard errors are reported in parentheses. $^{+}p<0.1$, * $p<0.05$, ** $p<0.01$, *** $p<0.001$.
\end{tablenotes}
\end{threeparttable}
\end{table}

\subsubsection*{Robustness check}

\begin{figure}
\centering 
\includegraphics[width=\textwidth]{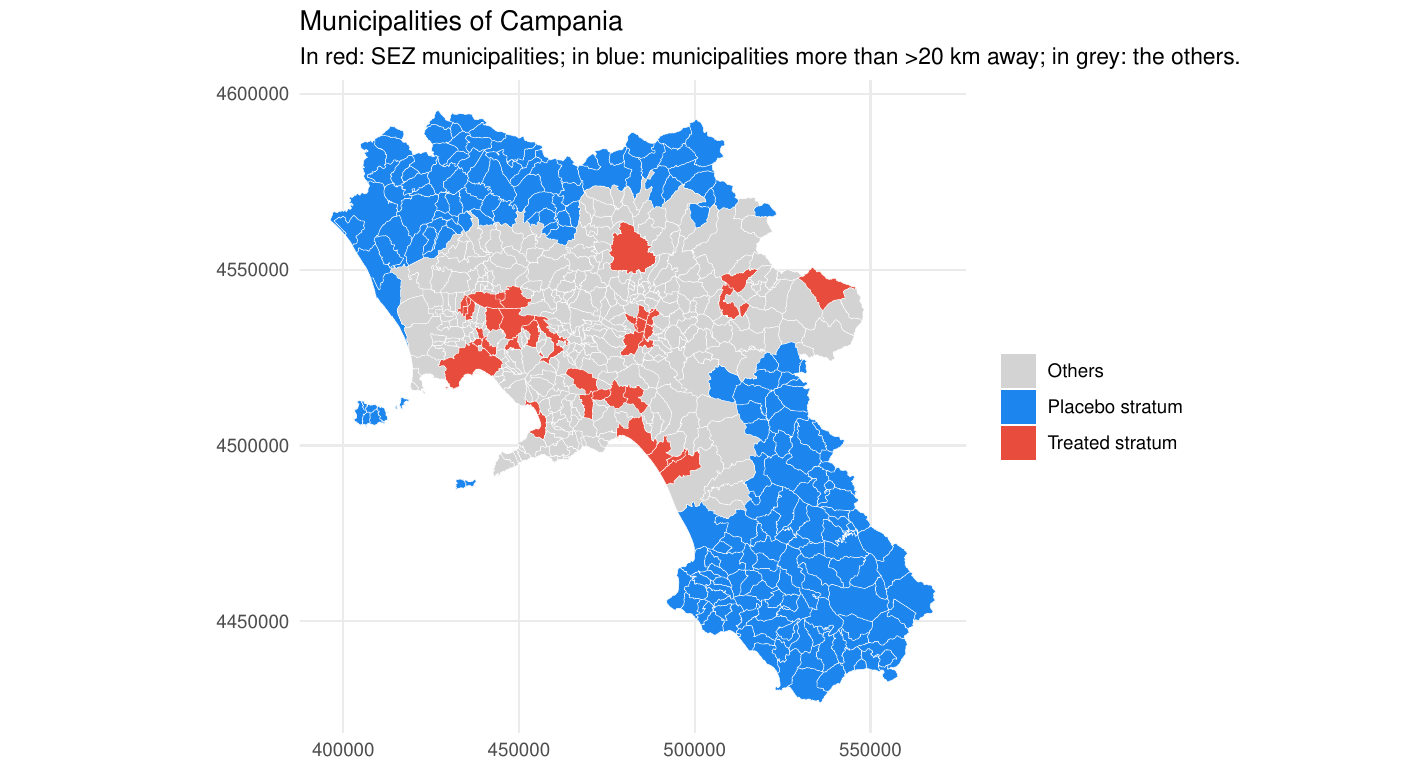} 
\caption{Map of municipalities in Campania used as control group in robustness check analysis (in blue), excluded municipalities are displayed in grey. } \label{fig:mapcampania}
\end{figure}

\begin{figure}
\centering
\includegraphics[width=\textwidth]{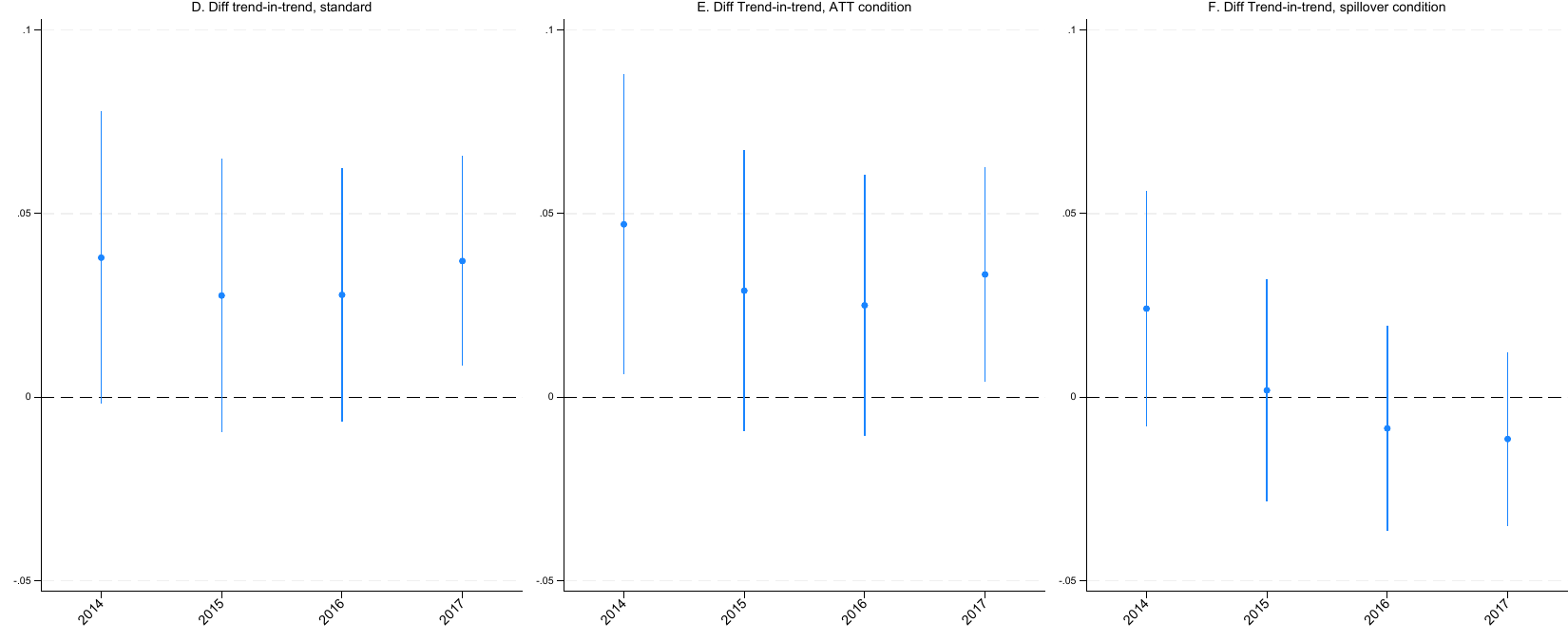} 
\caption{Tests on parallel trends and trend-in-trends assumptions in pre-policy years (2013-2017), considering distant Campania municipalities as the placebo stratum.} \label{fig:es_test_rob}
\end{figure}

As a robustness check for our empirical application, we redefine the placebo stratum using selected municipalities in Campania located outside the SEZ, rather than municipalities in Sicily. To mitigate concerns about geographical spillovers, we restrict the placebo municipalities to those at least 20 km away from treated municipalities, as illustrated in Figure \ref{fig:mapcampania}. Table~\ref{tab:triple_double} reports estimates from the three-way fixed-effects TD and DTD specifications with year-specific interactions. In the DTD model, the ATT is identified by the coefficient on the triple interaction \(T \times S \times G\). The estimate is statistically significant and implies an average \(9\%\) increase for treated firms. Over the same period, non-eligible firms exposed to spillovers and located in the SEZ area experienced a \(3.8\%\) increase (ASU), consistent with our main results. By contrast, the TD model---which implicitly rules out spillovers---yields a smaller effect of \(7.7\%\). Figure~\ref{fig:es_test_rob} reports tests of the parallel trend-in-trends assumptions using pre-policy data. Panels~D and~E suggest some slight deviations from the trend-in-trends assumption. Panel~F, instead, supports the parallel trend-in-trends condition relevant for identifying the ASU: the estimated \(3.8\%\) spillover effect coincides with the estimate obtained when Sicily is used as the placebo stratum.
    
\begin{table}[H]
\centering
\begin{threeparttable}
\caption{Three-way fixed effect TD and DTD estimates -- robustness check with distant municipalities in Campania as control group.}
\label{tab:triple_double}
\begin{tabular}{lcccc}
\toprule
 & \multicolumn{2}{c}{\textbf{Triple}} & \multicolumn{2}{c}{\textbf{Double-triple}} \\
\cmidrule(lr){2-3} \cmidrule(lr){4-5}
 & Coef. & S.E. & Coef. & S.E. \\
\midrule
S $\times$ 2013 & 0.040*   & (0.016) & 0.053*** & (0.016) \\
S $\times$ 2014 & 0.016    & (0.014) & 0.029*   & (0.014) \\
S $\times$ 2015 & 0.022+   & (0.013) & 0.035**  & (0.013) \\
S $\times$ 2016 & 0.029*   & (0.012) & 0.042*** & (0.012) \\
S $\times$ 2017 & 0.031**  & (0.011) & 0.045*** & (0.011) \\
S $\times$ 2018 & -0.003   & (0.010) & -0.003   & (0.010) \\
S $\times$ 2019 & 0.013    & (0.009) & 0.013    & (0.009) \\
S $\times$ 2020 & 0.011    & (0.007) & 0.011    & (0.007) \\
G $\times$ 2013 & -0.015   & (0.024) & -0.015   & (0.024) \\
G $\times$ 2014 & -0.014   & (0.023) & -0.013   & (0.023) \\
G $\times$ 2015 & -0.004   & (0.022) & -0.003   & (0.022) \\
G $\times$ 2016 & 0.002    & (0.021) & 0.002    & (0.021) \\
G $\times$ 2017 & -0.004   & (0.021) & -0.004   & (0.021) \\
G $\times$ 2018 & -0.052***& (0.010) & -0.052***& (0.010) \\
G $\times$ 2019 & -0.050***& (0.009) & -0.050***& (0.009) \\
G $\times$ 2020 & -0.022** & (0.007) & -0.022** & (0.007) \\
T $\times$ S $\times$ G & 0.077*** & (0.022) & 0.090*** & (0.022) \\
T $\times$ S $\times$ I &          &          & 0.038*** & (0.008) \\
\midrule
Num. Obs. & \multicolumn{2}{c}{335,222} & \multicolumn{2}{c}{335,222} \\
AIC       & \multicolumn{2}{c}{924,974.2} & \multicolumn{2}{c}{924,884.6} \\
BIC       & \multicolumn{2}{c}{3,877,384.9} & \multicolumn{2}{c}{3,877,284.7} \\
\bottomrule
\end{tabular}

\begin{tablenotes}
\scriptsize
\item \textit{Notes:} The Triple and Double-Triple columns report OLS estimates from three-way fixed effect specifications with firm fixed effects and interactions between year fixed effects and group indicators. $S$ identifies firms in Campania SEZ municipalities, $G$ identifies firms in targeted sectors, $I$ identifies firms in non-targeted sectors exposed to spillovers, and $T$ identifies policy periods (2018--2020).
\item Standard errors are reported in parentheses. $^{+}p<0.1$, * $p<0.05$, ** $p<0.01$, *** $p<0.001$.
\end{tablenotes}
\end{threeparttable}
\end{table}
\end{document}